\documentclass[a4paper,11pt]{article}
\usepackage[top=1.25in, bottom=1.25in, left=1.05in, right=1.05in]{geometry}
\usepackage{setspace}		

\usepackage[english]{babel}
\usepackage[utf8]{inputenc}
\usepackage[T1]{fontenc}

\usepackage{natbib}
\usepackage{xr}
\usepackage{xr-hyper} 					
\makeatother
\usepackage{enumitem}
\setlist{noitemsep}  
\usepackage[section]{placeins}		
\usepackage{amscd,amsmath,amssymb,amsfonts,amsthm,bbm,bm,latexsym,mathrsfs,mathtools}
\usepackage[subnum]{cases}	
\usepackage{dsfont}

\RequirePackage[colorlinks=true, citecolor=blue, urlcolor=blue]{hyperref}
\usepackage{booktabs}

\usepackage{xcolor,subfigure,epsfig,rotating,graphics,graphicx} 
\usepackage[width=1.2\linewidth, font={footnotesize}]{caption}

\usepackage{rotating,lscape}
\usepackage[flushleft]{threeparttable}
\usepackage{ragged2e,adjustbox}
\usepackage{longtable,float,array,multirow,colortbl}

\RequirePackage{hypernat}

\def\AS {ACPS}  

\def\II {\mathbb{I}}
\def\R {\mathds{R}}

\def\E {\mathbb{E}}
\newcommand*\abs[1]{\left|#1\right|}		
\newtheorem{theorem}{Theorem}

\newtheorem{definition}{Definition}
\theoremstyle{remark}
\newtheorem{rem}{Remark}
\theoremstyle{plain}
\newtheorem{exa}{Example}
\newcolumntype{C}[1]{>{\centering\arraybackslash}p{#1}}

\title{\vspace{-50pt} \textbf{Proper scoring rules for evaluating asymmetry in density forecasting}\thanks{\footnotesize The authors gratefully acknowledge Todd Clark, Michael McCracken, Massimiliano Marcellino, Barbara Rossi, Jonas Brehmer for their useful feedback. This paper is part of the research activities at the Centre for Applied Macroeconomics and Commodity Prices (CAMP) at the BI Norwegian Business School. This research used the SCSCF multiprocessor cluster system at Ca' Foscari University of Venice. Luca Rossini acknowledges financial support from the EU Horizon 2020 programme under the Marie Sk\l{}odowska-Curie scheme (grant agreement no. 796902). }}

\author{Matteo Iacopini\thanks{Ca' Foscari University of Venice, Italy. \color{blue}\texttt{matteo.iacopini@unive.it}} \and Francesco Ravazzolo\thanks{Free University of Bozen-Bolzano, Italy and CAMP, BI Norwegian Business School, Norway. \color{blue}\texttt{francesco.ravazzolo@unibz.it}} \and Luca Rossini\thanks{Queen Mary University of London, United Kingdom and Vrije Universiteit Amsterdam, The Netherlands. \color{blue}\texttt{l.rossini@qmul.ac.uk}}}

\begin{document}

\maketitle

\begin{abstract}
This paper proposes a novel asymmetric continuous probabilistic score (\AS) for evaluating and comparing density forecasts. It extends the proposed score  and defines a weighted version, which emphasizes regions of interest, such as the tails or the center of a variable's range. A test is also introduced to statistically compare the predictive ability of different forecasts. The \AS\, is of general use in any situation where the decision maker has asymmetric preferences in the evaluation of the forecasts. In an artificial experiment, the implications of varying the level of asymmetry in the \AS\, are illustrated. 
Then, the proposed score and test are applied to assess and compare density forecasts of  macroeconomic relevant datasets (US employment growth) and of commodity prices (oil and electricity prices) with particular focus on the recent COVID-19 crisis period.\\

\end{abstract}

\section{Introduction}
Macroeconomic forecasting has always been of pivotal importance for central bankers, policymakers and researchers. Nowadays, the vast majority of the research in macroeconomics and finance mainly focuses on the development and implementation of forecasting techniques that minimize the expected squared forecast error (\cite{gneiting2011making}). This approach is grounded on the implicit assumption of using a symmetric loss function in evaluating the accuracy of a forecast.

Despite being common practice, the use of symmetric loss functions in forecasting is unrealistic especially in policy institutions, where the policymakers could have a specific aversion to positive or negative deviations of a forecast from the target.
Consider a policymaker who is interested in forecasting employment. Suppose that, if the predicted employment rate drops below a given threshold, she will be forced to adopt new expansionary economic policy. It is highly likely that the policymaker is more averse to forecasts that give too high probability mass to the right part of the distribution of the employment rate (positive growth of employment), while she may be more relaxed with respect to forecasts that give too high probability mass to the left part of the distribution (negative or low growth of employment). Other examples relate to energy markets that have recently experienced negative prices. WTI oil prices collapsed to -37.63 US dollar for barrel in April 2020; German electricity prices have measured several negative prices with the introduction of renewable energy resources (RES). Producers would be more sensitive to prices below a threshold, up to zero if the marginal cost of production is zero, as it is the case of RES, than higher prices.
These examples call for the design of a more general class of loss functions and scoring rules that account for asymmetry, in order to guide the process of making and assessing forecasts. To the best of our knowledge, a measure that properly incorporates asymmetry in density forecasting evaluation does not exist in the literature.

The main goal of this paper is the proposal of novel and practical forecasting evaluation tools that can fill in this gap and answer the increasing demand from policymakers and central bankers.
We plan to achieve this result by introducing an innovative asymmetric scoring rule that is able to measure and evaluate heterogeneous aversion to different deviations of a density forecast from the target. We derive some properties of the new scoring rule and in particular demonstrate that it is a proper scoring rule. Moreover, we provide threshold- and quantile-weighted versions that allow to emphasize the performance of the forecast in regions of interest to the policymaker.

Within the literature on point forecasting, \cite{Christoffersen1996,Christoffersen1997} proposed some asymmetric loss functions. In the former paper, they studied the optimal prediction problem under general loss structures and characterized the optimal predictor under an asymmetric loss function, focusing on the \textit{LinEx} and the \textit{LinLin} asymmetric functions. In the latter paper, the authors provided an illustration of an asymmetric loss in the context of GARCH processes.

More recently, scholars have begun to empirically investigate the degree of loss function asymmetry of central banks and other international institutions. Among others, \cite{Elliott2005} and \cite{Patton2007} proposed formal methods to infer the degree of asymmetry of the loss function and to test the rationality of forecasts.
Within this stream of literature, \cite{Artis2001} found that IMF and OECD forecasts of the deficit of G7 countries are biased towards over or under-prediction relative to mean square error (MSE) forecasts. Regarding European institutions forecasts, \cite{Christodoulakis2008,Christodoulakis2009} found evidence of asymmetric loss.
In another study, \cite{Dovern2017} documented that the GDP growth forecasts made by professional forecasters tend to exhibit systematic errors, and tend to overestimate GDP growth.
Moreover, \cite{Boero2008} interpreted the tendency to over-predict GDP growth as a signal that policymakers exhibit greater fear of under-prediction than over-prediction, thus suggesting that their judgements are based on an asymmetric loss.
More recently, \cite{Tsuchiya2016} examined the asymmetry of the loss functions of the Japanese government, the IMF and private forecasters for Japanese growth and inflation forecasts.

In the framework of forecast combination, \cite{elliott2004optimal} showed that the optimal combination weights significantly differ under asymmetric loss functions and skewed error distributions as compared to those obtained with mean squared error loss.
Finally, \cite{demetrescu2019predictive} studied factor-augmented forecasting under asymmetric point loss function.

%

An alternative and more universal approach to forecasting is the provision of a predictive density, known as probabilistic or density forecasting (see \citet[][ch.8]{elliott2016economic}).
Two key aspects of density forecasts are the statistical compatibility between the forecasts and the realized observations (calibration) and the concentration of predictive distributions (sharpness). The aim of probabilistic forecasts is to maximize their sharpness, subject to calibration (\cite{gneiting2013combining}).
Density forecasting is more complex than point forecasting since the estimation problem requires to construct the whole predictive distribution, rather than a specific function thereof (e.g., mean or quantile).
Several reasons have been suggested for preferring density over point forecasts (e.g., \cite{elliott2016forecasting}).
First, point forecasting is often associated to the mean of a distribution and it is optimal for highly restricted loss functions, such as quadratic loss function, but inadequate for any prospective user having a different loss.
Moreover, the value of a point forecast can be increased by supplementing it with some measures of uncertainty and complete probability distributions over the outcomes provide useful information for making economic decisions; see, for example, \cite{Anscombe1968} and \cite{Zarnowitz1969} for early works and the discussions in \cite{gp00}, \cite{Timmermann2006} and \cite{gneiting2011making}. \cite{CCM2020} extends the application to tail risk nowcasts of economic activity.
Finally, in recursive forecasting with nonlinear models the full predictive density matters since the nonlinear effects typically depend not only on the conditional mean, but also on where future values occur in the set of possible outcomes.

A natural way to evaluate and compare competing density forecasts is the use of proper scoring rules, which assess calibration and sharpness simultaneously and encourage honest and careful forecasting.
Despite the wide literature on the class of proper scoring rules for probabilistic forecasts of categorical and binary variables (e.g., see \cite{savage1971elicitation}, \cite{schervish1989general}) the advances for continuous variables are more limited. Motivated by these facts, we aim at designing a novel asymmetric proper scoring rule to be used for evaluating density forecasts of continuous variables, which is the typical case in macroeconomics and finance exercises (e.g., predicting variables such as unemployment, inflation, log-returns, GDP growth, and realized volatility).

\cite{gneiting2007strictly} proposed the continuous rank probability score (CRPS) as a proper scoring rule for probabilistic forecasts of continuous variables, and more recently, \cite{gneiting2011comparing} extended the CRPS by introducing a threshold- and a quantile-weighted version (tCRPS and qCRPS, respectively).
These scoring rules give more emphasis to the performance of the density forecast in a selected \textit{region of the domain}, $B$, by assigning more weight to the deviations from the observations made in $B$. The major drawback of both the CRPS and its weighted versions is the symmetry of the underlying reward scheme, meaning that they assign equal reward to positive and negative deviations of a probabilistic forecast from the target.
This comes from the fact that the CRPS is built on the Brier score and inherits some of its properties, such as properness and symmetry.
Similarly, since both the weighted versions of the CRPS essentially consist in re-weighting the CRPS over the domain of the variable of interest, they inherit the symmetry of the latter.

\cite{winkler1994evaluating} did a first effort towards asymmetric scoring rules and proposed a general method for constructing asymmetric proper scoring rules starting from symmetric ones. However, this approach is limited to forecasting binary variables, and continuous variables were not investigated.

We address this issue and contribute to the literature on proper scoring rules for evaluating density forecasts by proposing a novel asymmetric proper scoring rule which assigns different penalties to positive and negative deviations from the true density.
The main contribution of this paper is twofold.
First, we define a new proper scoring rule which assigns an asymmetric penalty to deviations from the target density. Moreover, we provide a threshold- and quantile-weighted version of it and develop an adaptation of the Diebold-Mariano test to statistically compare the predictive ability of different forecasts.
Then, we compare the performance of the scores with the CRPS and its weighted versions.
Second, we use the proposed score to evaluate density forecasts in three relevant applications in  macroeconomics (US employment growth) and commodity prices (oil and electricity prices) with data updated to the COVID-19 crisis period. Variables have experienced large volatilities, with sizeable spikes and negative energy prices. As we discussed above, players might be more sensitive to some specific parts of the distribution of these series and we shed light on how to evaluate this asymmetry.

The key result of this paper is the provision of a tool able to account for the decision maker's preferences in the evaluation of density forecasts, both in terms of domain- and error-weighting schemes.
Domain-weighting gives heterogeneous emphasis to the performance on different regions, while the error-weighting asymmetrically rewards negative and positive deviations from the target value.
The proposed weighted asymmetric scoring rule combines the two schemes and allows to evaluate the performance of the forecasting density from both perspectives.
The rest of the paper is organized as follows.
Section~\ref{sec:asymmetric_scoring_rule} presents a novel asymmetric scoring rule for density forecasts, its extension to threshold- and quantile-weighted versions and a test to compare the predictive accuracy of different forecasts. Then Section~\ref{sec:simulations} discusses its main properties. It also illustrates a comparison with the (weighted) CRPS in simulated experiments.
Finally, Section~\ref{sec:application} provides different applications on forecasting US macroeconomic variables (employment rate) and commodity prices (oil and electricity prices).
The article closes with a discussion in Section~\ref{sec:conclusions}.

The \textsc{MATLAB} code for implementing the proposed scoring rules is available at:
\begin{center}
\url{https://github.com/matteoiacopini/acps}
\end{center}

\section{Asymmetric Proper Scoring rules for Density forecasting}   \label{sec:asymmetric_scoring_rule}

The evaluation and comparison of probabilistic forecasts typically relies on proper scoring rules. Informally, a scoring rule is a measure that summarises the goodness of a probabilistic forecast by combining the predictive distribution and the value that actually materializes.
One can think of it as a measure of distance between the probabilistic forecast and the actual value. 
We consider positively oriented scoring rules, therefore if probabilistic forecast $P_1$ obtains a higher score than $P_2$, this means that $P_1$ yields a more accurate forecast than $P_2$. Therefore, the score can be interpreted as a reward to be maximized.


In more formal terms, following the notation of \cite{gneiting2007strictly}, consider the problem of making probabilistic forecasts on a general sample space $\Omega$. Let $\mathcal{A}$ be a $\sigma$-algebra of subsets of $\Omega$, and let $\mathcal{P}$ be a convex class of probability measures on $(\Omega,\mathcal{A})$. A \textit{probabilistic forecast} is any probability measure $P \in \mathcal{P}$, such that 
$P : \Omega \to \bar{\R}$, where $\bar{\R} = [-\infty,+\infty]$ denotes the extended real line, is said to be $\mathcal{P}$-quasi-integrable if it is measurable with respect to $\mathcal{A}$ and is quasi-integrable with respect to all $P \in \mathcal{P}$ (see \cite{bauer2011measure}).
A \textit{scoring rule} is any extended real-valued function $S : \mathcal{P} \times \Omega \to \bar{\R}$ such that $S(P,\cdot)$ is $\mathcal{P}$-quasi-integrable for all $P \in \mathcal{P}$. 
In practice, if $P$ is the forecast density and the event $\omega$ materializes, then the forecaster's reward is $S(P,\omega)$.

In order to be effectively used in scientific forecasts evaluation, scoring rules have to be proper, meaning that they have to reward accurate forecasts.
Suppose the true density of the observations is $Q$ and denote 
the expected value of $S(P,\omega)$ under $Q(\omega)$ with
\begin{equation*}
S(P,Q) = \E_Q[S(P,\omega)] = \int_{\Omega} S(P,\omega) \: Q(d\omega),
\end{equation*}
then the scoring rule $S$ is \textit{strictly proper} if $S(Q,Q) \geq S(P,Q)$. The equality holds if and only if $P = Q$, thus implying that the forecaster has higher reward if she predicts $P=Q$.
If instead $S(Q,Q) \geq S(P,Q)$ for all $P$ and $Q$, then the scoring rule is said to be \textit{proper}.

The vast majority of the proper scoring rules proposed in the literature are symmetric, that is, they reward in the same way positive and negative deviations from the target. For example, suppose a forecast $P_1$ assigns too high probability mass to the right part of the domain (as compared to the true density) and a forecast $P_2$ assigns too high probability mass to the left part, by the same amount. If these forecasts are evaluated under a symmetric scoring rule, then they receive the same score.

A symmetric loss is unsatisfactory for many real world situations where the decision maker has a preference or aversion towards a particular kind of error.
%
%
We aim at filling in this gap by defining a new asymmetric proper scoring rule for continuous variables, which is suited for evaluation and comparison of density forecasts and penalises more either side of the deviation from the target.


\begin{definition}[Asymmetric Continuous Probability Score]
Let $c \in (0,1)$ represent the level of asymmetry, such that $c=0.5$ implies a symmetric loss, while $c < 0.5$ penalises more the left tail, and  $c > 0.5$ the right tail.
Let $P$ be the probabilistic forecast and $y$ the realized (ex-post) value.
We define the asymmetric continuous probability score (\AS) as
\begin{equation}
\small
\begin{aligned}
\AS & (P,y;c) = \int_{-\infty}^y \big( c^2 - P(u)^2 \big) \Big[ \frac{1}{(1-c)^2} \II(P(u) > c) + \frac{1}{c^2} \II(P(u) \leq c) \Big] \: du \\
 & + \int_y^{+\infty} \big( (1-c)^2 - (1-P(u))^2 \big) \Big[ \frac{1}{(1-c)^2} \II(P(u) > c) + \frac{1}{c^2} \II(P(u) \leq c) \Big] \: du.
\end{aligned}
\label{eq:asymmetric_score}
\end{equation}
\end{definition}

The following result shows the properness of our new score for every level of asymmetry.

\begin{theorem}[Properness]     \label{thm:properness_APSR}
The asymmetric scoring rule \AS\, defined in eq.~\eqref{eq:asymmetric_score} is strictly proper for any $c \in (0,1)$.
\end{theorem}
\begin{proof}
The strict properness derives from the fact that \AS\, can be obtained from the quadratic score for binary outcomes, which is strictly proper, via two transformations that preserve properness, see \cite{winkler1994evaluating} and \cite{matheson1976scoring}.
Specifically, let $p \in (0,1)$ be a probabilistic forecast of success in a binary experiment and let $S$ be the quadratic rule, that is
\begin{align*} \small
S(p) = \begin{cases}
S_1(p) = 1-(1-p)^2, & \text{ if success}, \\[5pt]
S_2(p) = 1-p^2,     & \text{ if failure}.
\end{cases}
\end{align*}
Notice that $S(p)$ is a strictly proper and symmetric scoring rule.
Following \cite{winkler1994evaluating}, one can obtain a strictly proper asymmetric scoring rule for binary outcomes via the transformation
\begin{align*} \small \hspace*{-4.5ex}
S_c^A(p) = \begin{cases}
\dfrac{S_1(p)-S_1(c)}{T(c)}, & \text{ if success}, \\[12pt]
\dfrac{S_2(p)-S_2(c)}{T(c)}, & \text{ if failure},
\end{cases}
\qquad
T(c) = \begin{cases}
S_1(1)-S_1(c), & \text{ if } p > c, \\[5pt]
S_2(0)-S_2(c), & \text{ if } p \leq c,
\end{cases}
\end{align*}
where $c \in (0,1)$ denotes the level of asymmetry.
Then, following  \cite{matheson1976scoring}, to obtain an asymmetric scoring rule for continuous variables, we 
assume that the subject assigns a probability distribution function $P(x)$ to a continuous variable of interest. Fix an arbitrary real number $u$ to divide the real line into two intervals, $I_1=\II(-\infty,u]$ and $I_2=\II(u, \infty)$, and define a success the event that $y$ falls in $I_1$. Since $P(u) \in (0,1)$ for any $u\in\R$, we can evaluate the binary scoring rule $S_c^A$ at $p = P(u)$, thus obtaining a different value $S_c^A(P(u))$ for each $u$.
Finally, the dependence of the scoring rule on the arbitrary value of $u$ is removed by integrating over all $u$, which yields eq.~\eqref{eq:asymmetric_score}.

Notice that one can obtain a different (strictly) proper asymmetric scoring rule as long as the baseline score is (strictly) proper.
\end{proof}

The integrals in eq.~\eqref{eq:asymmetric_score} can be numerically approximated by truncating the domain to $[u_{min},y]$ and $[y,u_{max}]$ such that
\begin{equation}
\small
\begin{aligned}
& \AS  (P,y;c) \approx \sum_{i=1}^{N} w_{2,i}^y \big( c^2 - P(u_{2,i}^y)^2 \big) \Big[ \frac{1}{(1-c)^2} \II(P(u_{2,i}^y) > c) + \frac{1}{c^2} \II(P(u_{2,i}^y) \leq c) \Big] \\
 & + \sum_{i=1}^{N} w_{1,i}^y \big( (1-c)^2 - (1-P(u_{1,i}^y))^2 \big) \Big[ \frac{1}{(1-c)^2} \II(P(u_{1,i}^y) > c) + \frac{1}{c^2} \II(P(u_{1,i}^y) \leq c) \Big],
\end{aligned}
\label{eq:asymmetric_score_approx}
\end{equation}
where $(w_{1,i}^y,u_{1,i}^y)_i$ and $(w_{2,i}^y,u_{2,i}^y)_i$, for $i=1,\ldots,N$, are the weights and locations of two Gaussian quadratures of $N$ points on $[y,u_{max}]$ and $[u_{min},y]$, respectively.

\begin{rem}
In Bayesian statistics it is current practice the use of predictive distributions, mostly in the form of Monte Carlo samples from posterior predictive distributions of quantities of interest.
The asymmetric scoring rule \AS\, can be easily computed using the output of a Markov chain Monte Carlo algorithm by approximating the predictive distribution via the empirical cumulative distribution function (empirical CDF) and using it as a probabilistic forecast $P$.
\end{rem}

To get an insight of the shape of the \AS\, for varying levels of asymmetry, $c$, we consider two examples: one with several probabilistic forecasts and the other with a fixed forecast.

\begin{figure}[!th]
\centering
\setlength{\abovecaptionskip}{0pt}
\hspace*{-4.5ex}
\begin{tabular}{ c c } 
\includegraphics[trim= 10mm 0mm 10mm 0mm,clip,height= 4.0cm, width= 7.0cm]{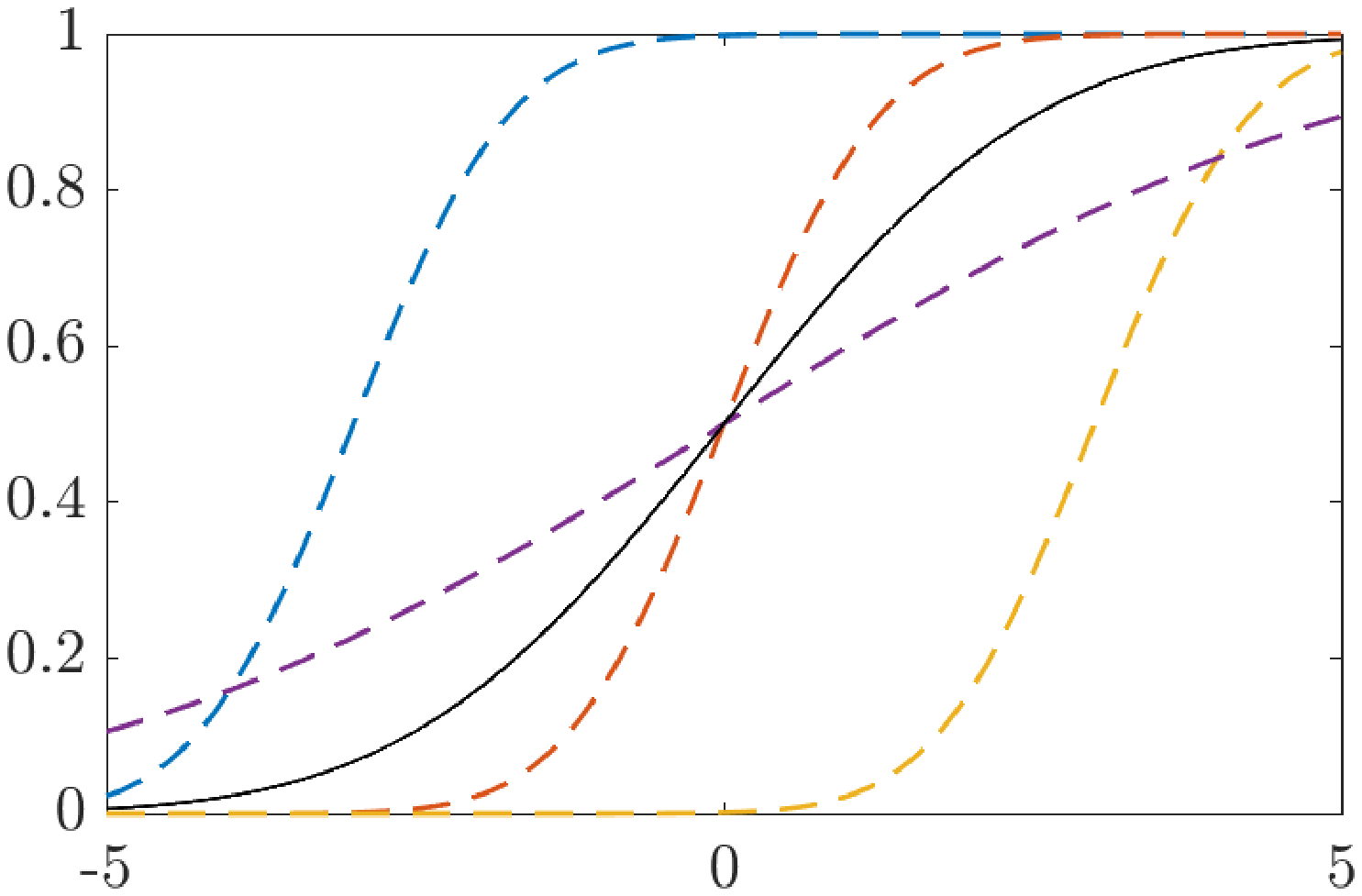} & 
\includegraphics[trim= 5mm 0mm 10mm 0mm,clip,height= 4.0cm, width= 7.0cm]{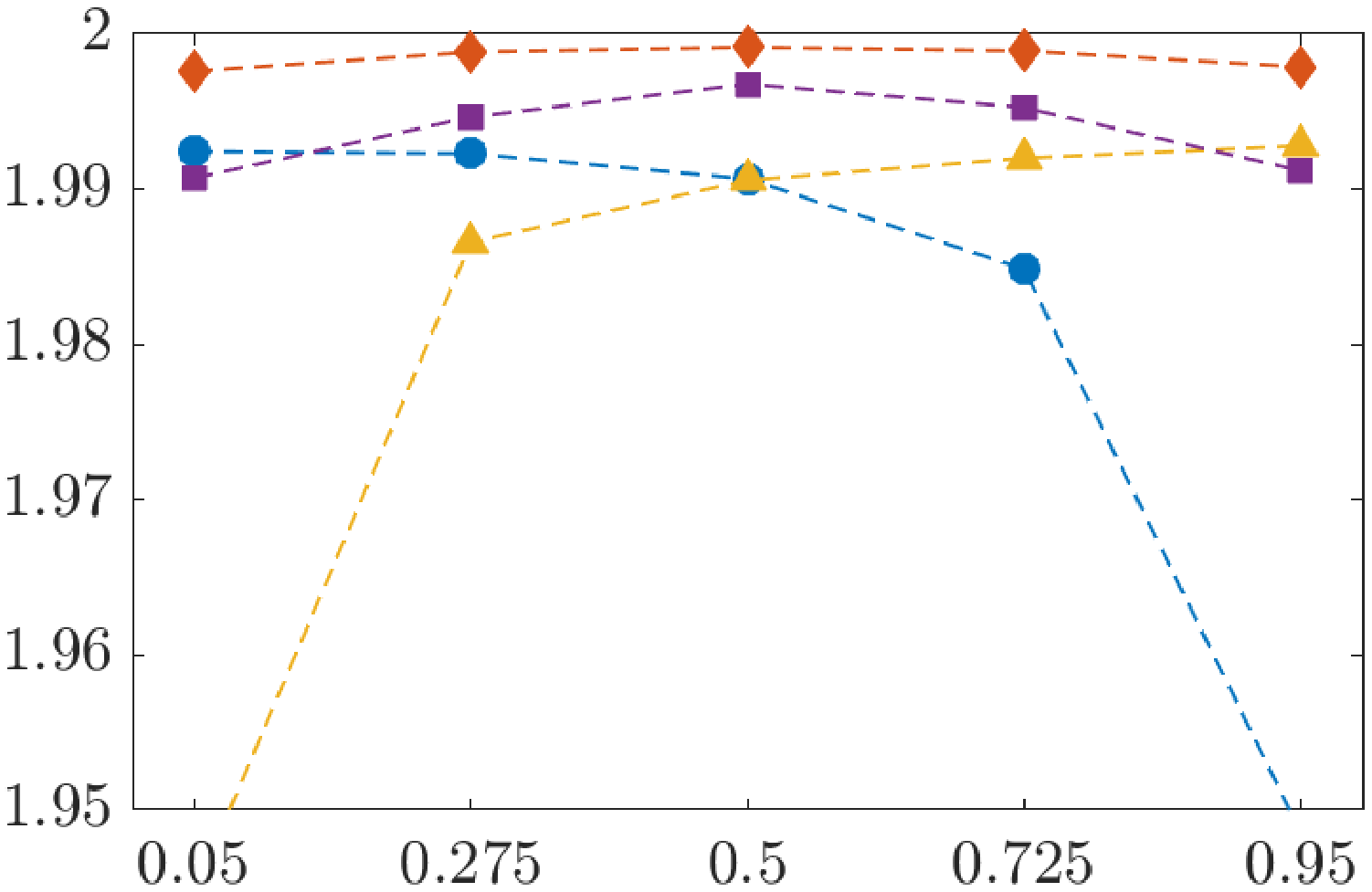}
\end{tabular}
\caption{Asymmetric scoring rule $\AS(P,y;c)$ for different forecasting densities $P$ and asymmetry level $c$. The observed value is fixed at $y=0$ and the true density is $\mathcal{N}(0,4)$.
Left panel: cumulative distribution functions of true density (solid, black) and forecasting densities: $\mathcal{N}(-3,1)$ (dashed, blue), $\mathcal{N}(0,1)$ (dashed, orange), $\mathcal{N}(3,1)$ (dashed, yellow), $\mathcal{N}(0,16)$ (dashed, purple).
Right panel: value of the asymmetric scoring rule $\AS(P,y;c)$ against the asymmetry level $c \in \{ 0.05, 0.275, 0.50, 0.725, 0.95 \}$, for each forecasting density (same colors as left panel).}
\label{fig:score_varyingP}
\end{figure}

\begin{exa}
Let us consider several Gaussian probabilistic forecasts $P$. In Fig.~\ref{fig:score_varyingP} we show the value of the score on a range of asymmetry values $c \in \{ 0.05, 0.275, 0.50, 0.725, 0.95 \}$, for a given observation $y$ whose true density is a standard Gaussian.
When the density forecast is Gaussian with the same mean as the target, the score is an inverse U-shaped function of the asymmetry level $c$. This is essentially due to the symmetry of the Gaussian distribution around its mean, since the probability mass in excess on the right tail is exactly equal to the mass lacking on the left one.
However, notice that a higher score is assigned to $\mathcal{N}(0,1)$, as compared to $\mathcal{N}(0,16)$.
Instead, the density forecasts $\mathcal{N}(-3,1)$ and $\mathcal{N}(3,1)$ receive a high penalty for high and small levels of $c$, respectively. This shows that values of $c$ close to $1$ heavily penalise forecasting densities that put more mass on the left part of the support as compared to the target, and conversely for values of $c$ close to $0$.
\end{exa}



\begin{exa}
Let us consider an alternative case when we keep fixed the probabilistic forecast to $\mathcal{N}(2,1)$ and inspect the value of the \AS \, for alternative target densities. As expected (see Fig.~\ref{fig:score_varyingY}), when the true density assigns more mass on the left part of the support as compared to the $\mathcal{N}(2,1)$, the forecast receives a very low score especially for $c$ close to $0$. Conversely, when the underlying true density is $\mathcal{N}(3,1)$ the forecast receives a higher reward for $c=0.05$, since its CDF is basically a left-shifted version of the target.
\end{exa}


\begin{figure}[!th]
\centering
\setlength{\abovecaptionskip}{0pt}
\hspace*{-4.5ex}
\begin{tabular}{ c c } 
\includegraphics[trim= 10mm 0mm 10mm 0mm,clip,height= 4.0cm, width= 7.0cm]{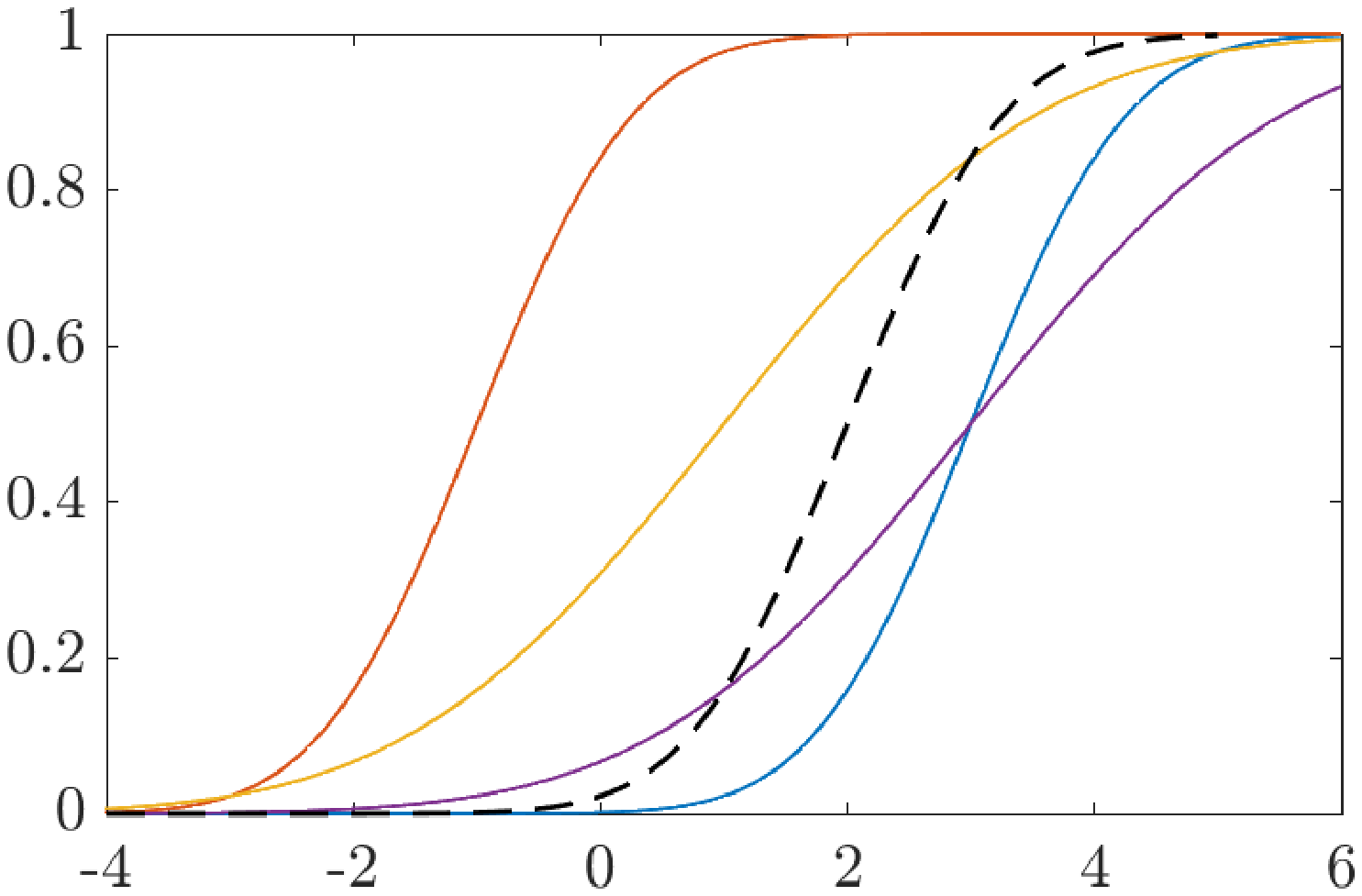} & 
\includegraphics[trim= 5mm 0mm 10mm 0mm,clip,height= 4.0cm, width= 7.0cm]{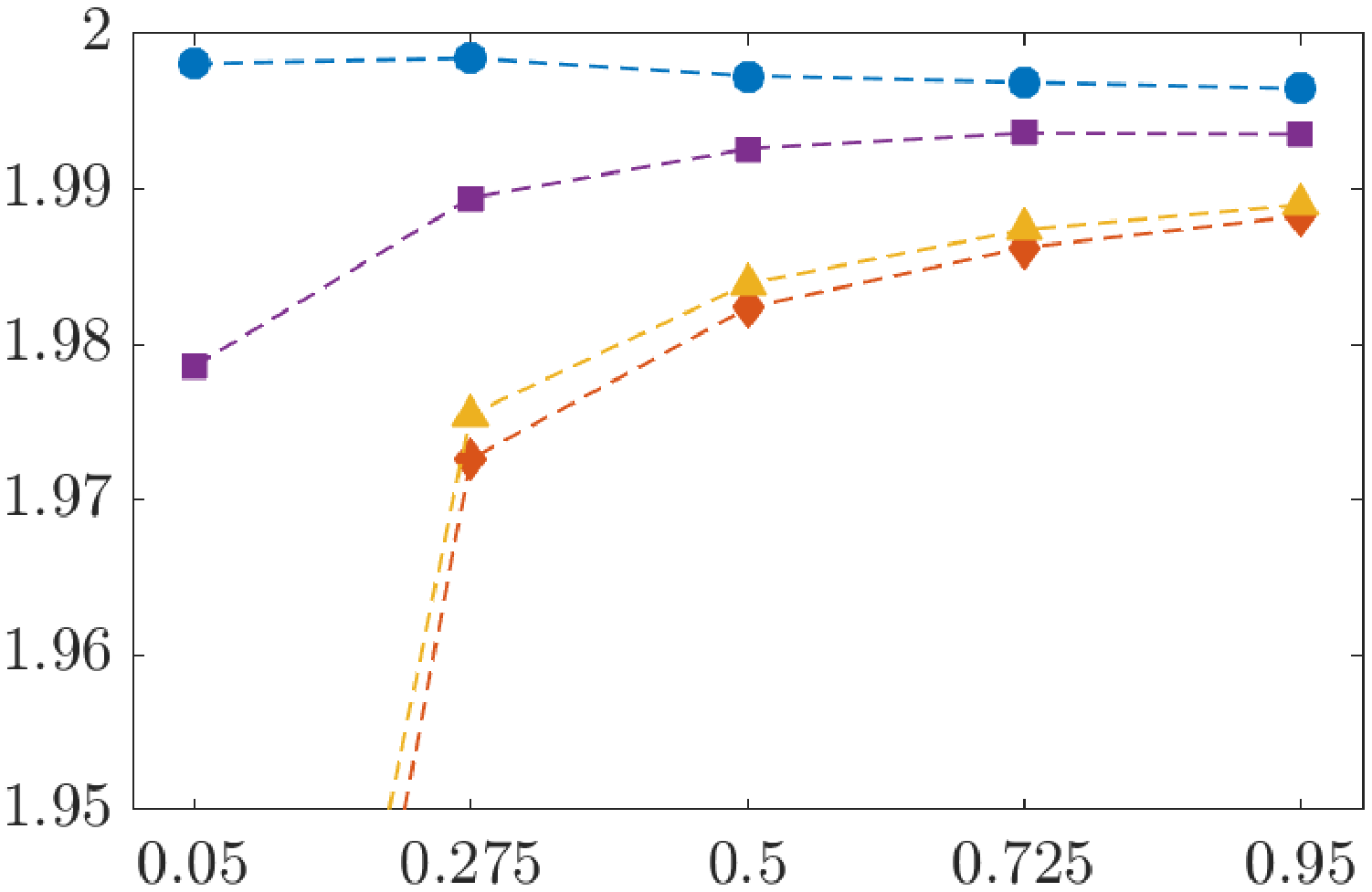} 
\end{tabular}
\caption{Asymmetric scoring rule $\AS(P,y;c)$ for different observed values $y$ and asymmetry level $c$. The forecasting density is fixed at $P = \mathcal{N}(2,1)$.
Left panel: cumulative distribution functions of the forecasting density (dashed, black) and of observation densities: $\mathcal{N}(3,1)$ (solid, blue), $\mathcal{N}(-1,1)$ (solid, orange), $\mathcal{N}(1,4)$ (solid, yellow), $\mathcal{N}(3,4)$ (solid, purple).
Right panel: value of the asymmetric scoring rule $\AS(P,y;c)$ against the asymmetry level $c \in \{ 0.05, 0.275, 0.50, 0.725, 0.95 \}$, for each observation density (same colors as left panel).}
\label{fig:score_varyingY}
\end{figure}


\subsection{Threshold and quantile-weighted versions}

In addition to asymmetric preferences towards under- or overestimation, a decision maker is usually concerned with a precise forecast in a specific range of all possible values.
Therefore, it is important to have a tool that allows to assign heterogeneous weights to various regions of the set of possible values of the variable. This calls for a scoring rule able to account for both error-weighting, i.e. asymmetric preferences and domain-weighting of density forecasts.

\cite{gneiting2011comparing} modified the CRPS by re-weighting the loss according to a user-specified weight function, which allows to select the regions where the decision-maker has greater concern.
By exploiting the representation of the CRPS in terms of quantile functions, they define a threshold-weighted (tCRPS) and quantile-weighted (qCRPS) score functions as follows
\begin{align}
tCRPS(P,y) & = \int_{-\infty}^{+\infty} \abs{P(z) - \II(y \leq z)}^2 w(z) \: dz, \\
qCRPS(P,y) & = \int_0^1 2\big(  \II(y \leq P^{-1}(\alpha)) -\alpha \big) (P^{-1}(\alpha)-y) v(\alpha) \: d\alpha,
\end{align}
where $w(z) \geq 0$ and $v(\alpha) \geq 0$ are the weight functions and level $\alpha \in (0,1)$.
Table~\ref{tab:weights} reports some examples of weighting functions for the case of real-valued variables of interest; notice that the uniform weight, $w(z)=1$ and $v(\alpha)=1$, leads to the standard CRPS. See \cite{ltrg2017} for discussion and applications of these scoring rules.

%
\begin{table}[!th]
\centering
\caption{Examples of weight functions for threshold-weighted and quantile-weighted CRPS, and variables supported on the real line.
$\phi,\Phi$ denote the probability density and cumulative distribution functions of the standard Normal distribution, respectively, with $x \in \R$ and $\alpha \in (0,1)$.}
\label{tab:weights}
\footnotesize
\begin{tabular}{lll}
\hline
Emphasis   & Threshold weight function   & Quantile weight function \\
\hline
uniform    & $w(x) = 1$ 					& $v(\alpha) = 1$ \\
center     & $w(x) = \phi(x)$ 			& $v(\alpha) = \alpha(1-\alpha)$ \\ 
tails      & $w(x) = 1-\phi(x)/\phi(0)$ 	& $v(\alpha) = (2\alpha-1)^2$ \\
right tail & $w(x) = \Phi(x)$ 			& $v(\alpha) = \alpha^2$ \\
left tail  & $w(x) = 1-\Phi(x)$ 			& $v(\alpha) = (1-\alpha)^2$ \\
\hline
\end{tabular}
\end{table}

%
%

The definition of \AS\, in \eqref{eq:asymmetric_score} can be modified to address this issue and obtain a threshold-weighted and a quantile-weighted asymmetric scoring rule, as follows.

\begin{definition}[Threshold-weighted \AS]
Let $G(du)$ be a positive measure\footnote{Notice that $G(du)$ is not required to be a probability measure.}.
We define the threshold-weighted asymmetric continuous probability score (t\AS), as
\begin{equation}
\small
\begin{aligned}
t\AS & (P,y;c) = \int_{-\infty}^y \big( c^2 - P(u)^2 \big) \Big[ \frac{1}{(1-c)^2} \II(P(u) > c) + \frac{1}{c^2} \II(P(u) \leq c) \Big] \: G(du) \\
 & + \int_y^{+\infty} \big( (1-c)^2 - (1-P(u))^2 \big) \Big[ \frac{1}{(1-c)^2} \II(P(u) > c) + \frac{1}{c^2} \II(P(u) \leq c) \Big] \: G(du),
\end{aligned}
\label{eq:asymmetric_score_threshold_weighted}
\end{equation}
where $c\in (0,1)$ is the level of asymmetry and $P$ is the probabilistic forecast and $y$ the value that materializes.
\end{definition}

\begin{definition}[Quantile-weighted \AS]
Let $p(u)$ denote the probability density function of $P(u)$ and let $P^{-1}(\alpha)$ be the corresponding quantile function at $\alpha \in [0,1]$.
Let $V(d\alpha)$ be a positive measure on the unit interval.
We define the quantile-weighted asymmetric continuous probability score (q\AS), as
\begin{equation}
\small
\begin{aligned}
q\AS & (P,y;c) = \int_{0}^{P(y)} \!\!\! \big( c^2 - \alpha^2 \big) \Big[ \frac{1}{(1-c)^2} \II(\alpha > c) + \frac{1}{c^2} \II(\alpha \leq c) \Big] \, \frac{1}{p(P^{-1}(\alpha))} \: V(d\alpha) \\
 & + \int_{P(y)}^{1} \!\!\!\!\!\! \big( (1-c)^2 - (1-\alpha)^2 \big) \Big[ \frac{1}{(1-c)^2} \II(\alpha > c) + \frac{1}{c^2} \II(\alpha \leq c) \Big] \, \frac{1}{p(P^{-1}(\alpha))} \: V(d\alpha).
\end{aligned}
\label{eq:asymmetric_score_quantile_weighted}
\end{equation}
\end{definition}

As stated for \AS, we can provide evidence of the properness of the two novel scores defined in eq.~\eqref{eq:asymmetric_score_threshold_weighted} and eq.~\eqref{eq:asymmetric_score_quantile_weighted}.

\begin{theorem}[Properness of $t\AS$, $q\AS$]
For any $c \in (0,1)$, it holds:
\begin{enumerate}[label=\alph*)]
\item the threshold-weighted asymmetric continuous probability score $t\AS$ in eq.~\eqref{eq:asymmetric_score_threshold_weighted} is strictly proper;
\item the quantile-weighted asymmetric continuous probability score $q\AS$ in eq.~\eqref{eq:asymmetric_score_quantile_weighted} is strictly proper.
\end{enumerate}
\end{theorem}
\begin{proof}
The result follows from Theorem~\ref{thm:properness_APSR} and \cite{matheson1976scoring}.
\end{proof}

Both t\AS\, and q\AS\, can be computed by approximating eq.~\eqref{eq:asymmetric_score_threshold_weighted} and eq.~\eqref{eq:asymmetric_score_quantile_weighted} in a way analogous to eq.~\eqref{eq:asymmetric_score_approx}.
The main advantage of the t\AS\, and q\AS\, consists in the ability to consider two levels of asymmetry: in terms of the loss at each point, and over different regions of the domain.
This is fundamental to answer the need of the decision maker who is concerned with the performance of the forecast in a given interval of possible values (e.g., the right tail) and who has an aversion to particular deviations from the target (e.g., averse to underestimation).

Tab.~\ref{tab:CRPS_vs_asymmetric_score} provides a summary of some key differences between the CRPS and \AS, and the corresponding weighted versions.

%

\begin{table}[!th]
\centering
\caption{Examples of scoring rules for evaluating density forecasts.}
\label{tab:CRPS_vs_asymmetric_score}
\small
\setlength{\tabcolsep}{6pt}
\renewcommand*{\arraystretch}{1.15}
\begin{tabular}{ c l | c c }
\hline \hline
 & & \multicolumn{2}{c}{\textbf{Domain}} \\
 & & uniform & weighted \\
\hline
\multirow{2}{*}{\begin{rotate}{90} \hspace*{-12pt} \textbf{Loss} \end{rotate}} & symmetric  & CRPS & tCRPS, qCRPS \\
 & asymmetric & \AS & t\AS, q\AS \\
\hline
\end{tabular}
\end{table}


\begin{rem}[Multivariate case]
The proposed asymmetric scores can be easily generalized to multivariate settings.
To this aim, denote with $\mathcal{Q}$ the class of the Borel probability measures on $\R^n$ and let $F \in \mathcal{Q}$ be a probabilistic forecast identified via its cumulative distribution function, $P$.
Let $c \in (0,1)$ represent the level of asymmetry and denote with $\mathbf{y} = (y_1,\ldots,y_n)'$ the multivariate value that materializes.
The multivariate version of the asymmetric continuous probability score is defined as
\begin{equation}
\small
\begin{aligned}
& \AS (P,\mathbf{y};c) = \\
 & = \int_{-\infty}^{y_n} \cdots \int_{-\infty}^{y_1} \big( c^2 - P(\mathbf{u})^2 \big) \Big[ \frac{1}{(1-c)^2} \II(P(\mathbf{u}) > c) + \frac{1}{c^2} \II(P(\mathbf{u}) \leq c) \Big] \: d\mathbf{u} \\
 & + \int_{y_n}^{+\infty} \cdots \int_{y_1}^{+\infty} \big( (1-c)^2 - (1-P(\mathbf{u}))^2 \big) \Big[ \frac{1}{(1-c)^2} \II(P(\mathbf{u}) > c) + \frac{1}{c^2} \II(P(\mathbf{u}) \leq c) \Big] \: d\mathbf{u},
\end{aligned}
\label{eq:asymmetric_score_multivariate}
\end{equation}
where $d\mathbf{u} = du_1 \cdots du_n$.
Moreover, one can define a multivariate threshold-weighted \AS \, by substituting the product of Lebesgue measures in eq.~\eqref{eq:asymmetric_score_multivariate} with a positive measure $G(d\mathbf{u})$ on $\R^n$.
\end{rem}

In the literature on asymmetric point forecasting measures, the choice of the shape parameter(s) of the loss function has been of interest especially over the last decade.
Some works, such as \cite{Christoffersen1996} and \cite{demetrescu2019predictive} performed an empirical exercise to rank competing forecasting models under asymmetric point forecast measure using a grid of asymmetry values.

Instead, \cite{Elliott2005,elliott2008biases} and \cite{Patton2007} introduced procedures for inferring the value of the shape parameters of the forecaster's loss function. Assuming a collection of time series of forecasts of a given quantity of interest is available, they treat the loss function parameters as variable to be estimated and look for the values that would be most consistent with forecast rationality.
This approach is appealing since the loss function parameters may provide information about the forecaster's objectives. However, the main drawback of these approaches is that they rely on the availability of a time series of observed forecasts from relevant decision-makers (e.g., the IMF, the OECD, or central bankers). Therefore, the application is case specific for those data.

\subsection{Testing predictive ability}\label{sec:test}
When forecasts from multiple models are available, there is the need for statistical tools, such as tests, for assessing whether different forecasts are equally good.
In the context of point forecasts, the Diebold-Mariano (DM) test is the most frequently used test for equal forecast performance. Essentially, it is based on the loss differential, defined as 
\begin{equation*}
d_t = L(e_{1,t}) - L(e_{2,t}),
\end{equation*}
where $e_{j,t} = \hat{y}_{j,t} - y_t$ is the forecast error of model $j=1,2$ at time $t=1,\ldots,T$, $\hat{y}_{j,t}$ is the point forecast of model $j$, $y_t$ is the true value, and $L(\cdot)$ is a given loss function. The null hypothesis of equal accuracy in forecasting is $H_0: \E[d_t] = 0$ for all $t$, versus the alternative $H_1: \E[d_t] \neq 0$. It can be shown that, if the loss differential series is (i) covariance stationary, and (ii) has short memory (see e.g. \cite{mccracken2020diverging}), then under the null hypothesis
\begin{equation*}
\frac{\sqrt{T} \bar{d}}{\sqrt{2\pi f_d(0)}} \to \mathcal{N}(0,1),
\end{equation*}
where $\bar{d}$ and $f_d(0)$ are the sample mean and the spectral density (at frequency 0) of the loss differential.
Recently, \cite{mccracken2020diverging} found that the slow decay of the loss differential series is the most frequent problem that hampers the use of the Diebold-Mariano test in real data economic applications.

The density forecasting approach requires some adaptations of the Diebold-Mariano test, since the forecast is an infinite dimensional object $P$.

\begin{rem}[Modified DM test]
To test the null hypothesis of equal accuracy of two competing models in a density forecasting approach, we modify the definition of the loss differential as follows.
First, consider a proper scoring rule $S$, such as the \AS \, or the CRPS, and denote the associated loss with $S^*(y,P) = -S(y,P)$. Then, the loss differential is defined as
\begin{equation}
d_t^* = S^*(y_t,P_{1,t}) - S^*(y_t,P_{2,t}).
\end{equation}
Notice that the series $d_t^*$ has the same interpretation as $d_t$ in the original DM test, and following the same theoretical arguments one can prove that, under the null hypothesis $H_0: \E[d_t^*] = 0$ for each $t$, one has
\begin{equation}
\frac{\sqrt{T} \bar{d}^*}{\sqrt{2\pi f_{d^*}(0)}} \to \mathcal{N}(0,1),
\end{equation}
where $\bar{d}^*$ and $f_{d^*}(0)$ are the equivalent of $\bar{d}$ and $f_d(0)$ for $d_t^*$.
\end{rem}

\section{Illustrations and comparison with weighted CRPS}  \label{sec:simulations}
This section investigates the performance of the proposed asymmetric scoring rule and compares it with the CRPS.
In order to assess the good performance of our measure, we consider different target densities: (i) Gaussian, (ii) Student-t, (iii) Gamma, (iv) Beta. This range includes families of distributions with different support ($\R$, $\R_+$ and $[0,1]$), skewed and with fat tails.
For the asymmetric scoring rule \AS \, we use varying levels of asymmetry, corresponding to $c \in \{ 0.05, 0.275, 0.50, 0.725, 0.95 \}$. Recall that $c=0.50$ implies a symmetric loss.


\begin{figure}[!th]
\centering
\renewcommand*{\arraystretch}{0.9}
\begin{minipage}[c]{0.46\linewidth}
\includegraphics[trim= 10mm 0mm 10mm 0mm,clip,height= 4.0cm, width= 6.0cm]{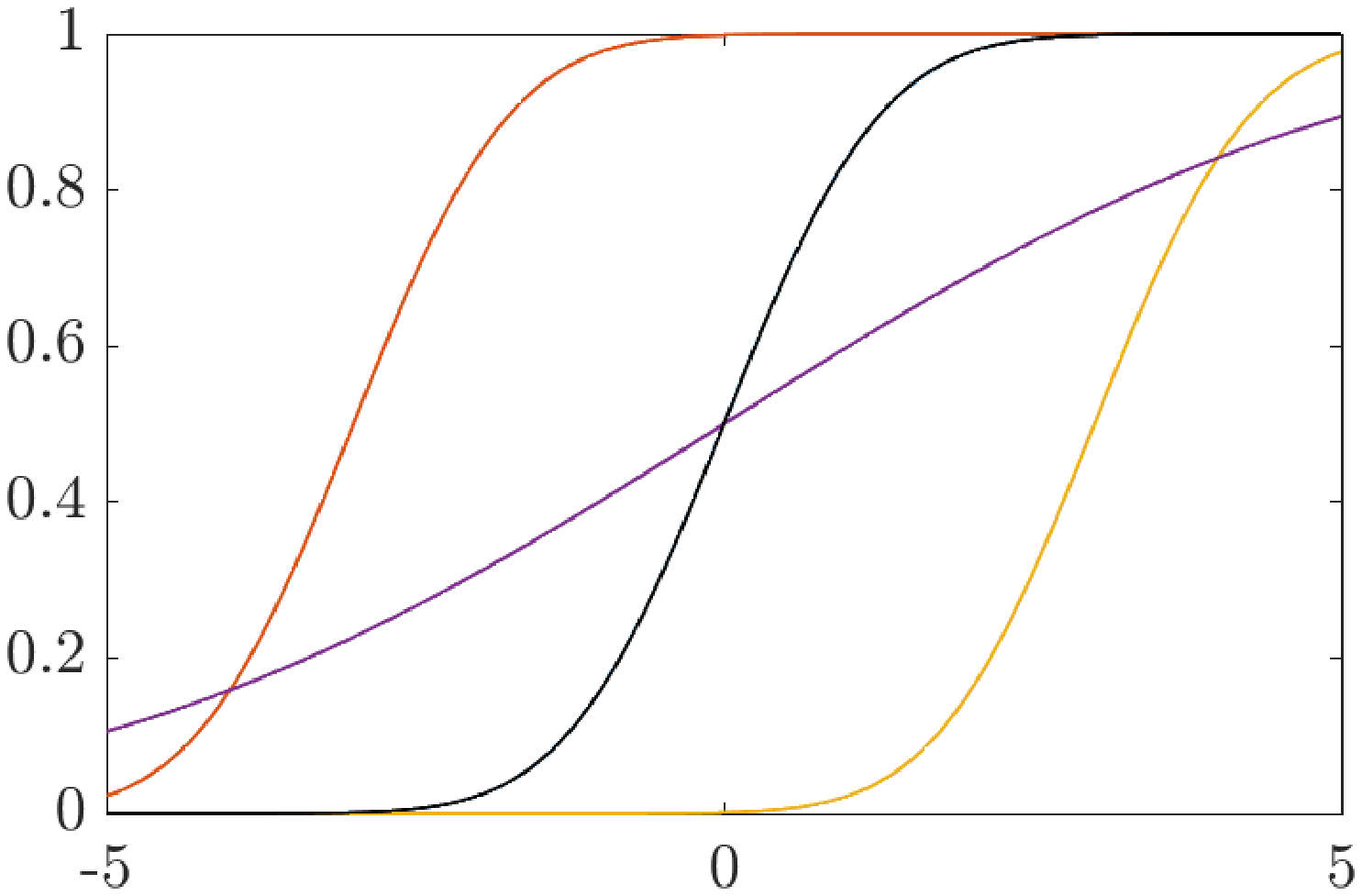}
\end{minipage}
\begin{minipage}[c]{0.46\linewidth}
\small
\begin{adjustbox}{width=1.1\textwidth, center=\textwidth}
\begin{tabular}{ l | *{4}{C{36pt}} } 
\hline \hline
& \multicolumn{4}{c}{Forecasting density} \\
& $\mathcal{N}(0,1)$& $\mathcal{N}(-3,1)$& $\mathcal{N}(3,1)$& $\mathcal{N}(0,16)$\\ 
\hline 
CRPS             & \textbf{1} & 4 & 3 & 2  \\ 
\AS$(\cdot,\cdot;0.05)$  & \textbf{1} & 2 & 4 & 3  \\ 
\AS$(\cdot,\cdot;0.275)$ & \textbf{1} & 3 & 4 & 2  \\ 
\AS$(\cdot,\cdot;0.5)$   & \textbf{1} & 4 & 3 & 2  \\ 
\AS$(\cdot,\cdot;0.725)$ & \textbf{1} & 4 & 3 & 2  \\ 
\AS$(\cdot,\cdot;0.95)$  & \textbf{1} & 4 & 2 & 3  \\ 
\hline 
\end{tabular}
\end{adjustbox}
\end{minipage}
\caption{Ranking of probabilistic forecasts. Results from $S=1$ simulation of $N=100$ observations. Density estimated with $M=500$ draws from forecasting distribution. Target is $\mathcal{N}(0,1)$ (black), forecasting densities are: $\mathcal{N}(0,1)$ (blue), $\mathcal{N}(-3,1)$ (orange), $\mathcal{N}(3,1)$ (yellow), $\mathcal{N}(0,16)$ (purple).}
\label{tab:test_NormNorm}
\end{figure}

Fig.~\ref{tab:test_NormNorm} and Fig.~\ref{tab:test_tt} provide graphical evidence of the properness of the \AS \, in two cases, with a Gaussian and a Student-t target, respectively.
Both figures show that the \AS \, rewards the forecast density which corresponds to the ground truth, for all levels of asymmetry.
In addition, we find that the ranking of the competing probabilistic forecasts changes according to the value of $c$, due to the different penalty assigned to asymmetric deviations from the target.

\begin{figure}[!th]
\centering
\renewcommand*{\arraystretch}{0.9}
\begin{minipage}[c]{0.46\linewidth}
\includegraphics[trim= 10mm 0mm 10mm 0mm,clip,height= 4.0cm, width= 6.0cm]{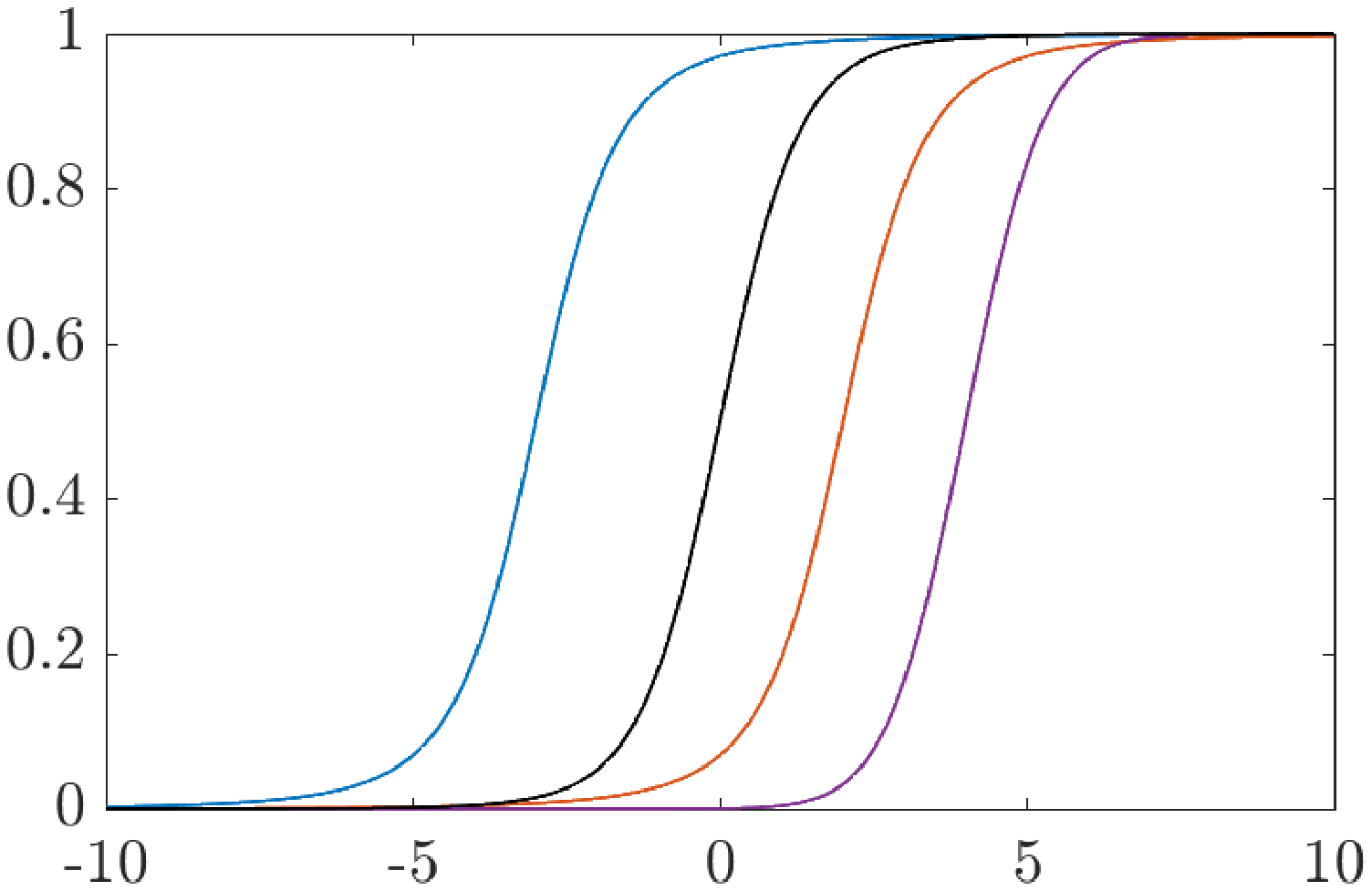}
\end{minipage}
\begin{minipage}[c]{0.46\linewidth}
\small
\begin{adjustbox}{width=1.1\textwidth, center=\textwidth}
\begin{tabular}{ l | *{4}{C{40pt}} } 
\hline \hline
& \multicolumn{4}{c}{Forecasting density} \\
& $t(-3,1,3)$& $t(2,1,3)$& $t(0,1,5)$& $t(4,1,15)$\\ 
\hline 
CRPS             & 3 & 2 & \textbf{1} & 4  \\ 
\AS$(\cdot,\cdot;0.05)$  & 3 & 2 & \textbf{1} & 4  \\ 
\AS$(\cdot,\cdot;0.275)$ & 3 & 2 & \textbf{1} & 4  \\ 
\AS$(\cdot,\cdot;0.5)$   & 3 & 2 & \textbf{1} & 4  \\ 
\AS$(\cdot,\cdot;0.725)$ & 4 & 2 & \textbf{1} & 3  \\ 
\AS$(\cdot,\cdot;0.95)$  & 4 & 2 & \textbf{1} & 3  \\ 
\hline 
\end{tabular}
\end{adjustbox}
\end{minipage}
\caption{Ranking of probabilistic forecasts. Results from $S=1$ simulation of $N=100$ observations. Density estimated with $M=500$ draws from forecasting distribution. Target is $t(0,1,5)$ (black), forecasting densities are: $t(-3,1,3)$ (blue), $t(2,1,3)$ (orange), $t(0,1,5)$ (yellow), $t(4,1,15)$ (purple).}
\label{tab:test_tt}
\end{figure}

To investigate further this aspect, Fig.~\ref{tab:test_NormNorm2} presents the ranking of forecasts when none of the candidates corresponds to the true density, which is $\mathcal{N}(2,4)$.
The CRPS indicates $\mathcal{N}(3,1)$ as the best forecast, as does the \AS \, for values of $c$ around $0.5$. However, when the \AS \, assigns more weight to the asymmetric loss, that is for $c$ close to the boundary of $(0,1)$, the ranking significantly changes.
For $c=0.05$, that is when great importance is given to underestimation of the target, the $\mathcal{N}(0,1)$ is preferred, while $\mathcal{N}(0,16)$ is the best for the opposite case, when $c=0.95$.


\begin{figure}[!th]
\centering
\renewcommand*{\arraystretch}{0.9}
\begin{minipage}[c]{0.46\linewidth}
\includegraphics[trim= 10mm 0mm 10mm 0mm,clip,height= 4.0cm, width= 6.0cm]{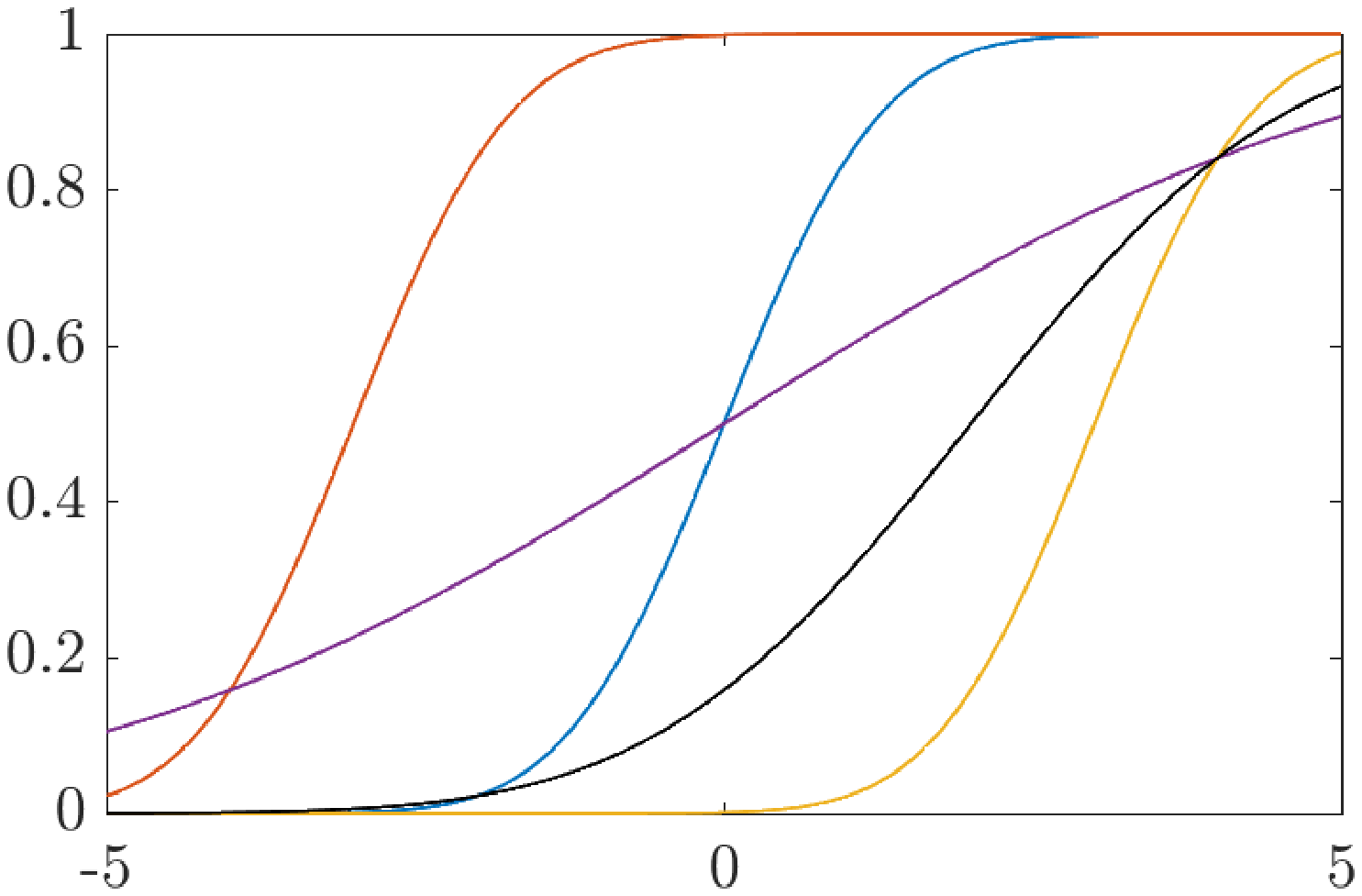}
\end{minipage}
\begin{minipage}[c]{0.46\linewidth}
\small
\begin{adjustbox}{width=1.1\textwidth, center=\textwidth}
\begin{tabular}{ l | *{4}{C{36pt}} } 
\hline \hline
& \multicolumn{4}{c}{Forecasting density} \\
& $\mathcal{N}(0,1)$& $\mathcal{N}(-3,1)$& $\mathcal{N}(3,1)$& $\mathcal{N}(0,16)$\\ 
\hline
CRPS             & 3 & 4 & \textbf{1} & 2  \\ 
\AS$(\cdot,\cdot;0.05)$  & \textbf{1} & 3 & 4 & 2  \\ 
\AS$(\cdot,\cdot;0.275)$ & 2 & 4 & \textbf{1} & 3  \\ 
\AS$(\cdot,\cdot;0.5)$   & 3 & 4 & \textbf{1} & 2  \\ 
\AS$(\cdot,\cdot;0.725)$ & 3 & 4 & \textbf{1} & 2  \\ 
\AS$(\cdot,\cdot;0.95)$  & 3 & 4 & 2 & \textbf{1}  \\ 
\hline 
\end{tabular}
\end{adjustbox}
\end{minipage}
\caption{Ranking of probabilistic forecasts. Results from $S=1$ simulation of $N=100$ observations. Density estimated with $M=500$ draws from forecasting distribution. Target is $\mathcal{N}(2,4)$ (black), forecasting densities are: $\mathcal{N}(0,1)$ (blue), $\mathcal{N}(-3,1)$ (orange), $\mathcal{N}(3,1)$ (yellow), $\mathcal{N}(0,16)$ (purple).}
\label{tab:test_NormNorm2}
\end{figure}

Many economic and financial variables in levels are inherently positive (e.g. GDP, volatility) or take values on a bounded interval (e.g., interest rate, unemployment rate). To account for these cases, we investigate the performance of the \AS \, in simulated experiments where the  target density is either Gamma or Beta.


\begin{figure}[!th]
\centering
\renewcommand*{\arraystretch}{0.9}
\begin{minipage}[c]{0.46\linewidth}
\includegraphics[trim= 10mm 0mm 10mm 0mm,clip,height= 4.0cm, width= 6.0cm]{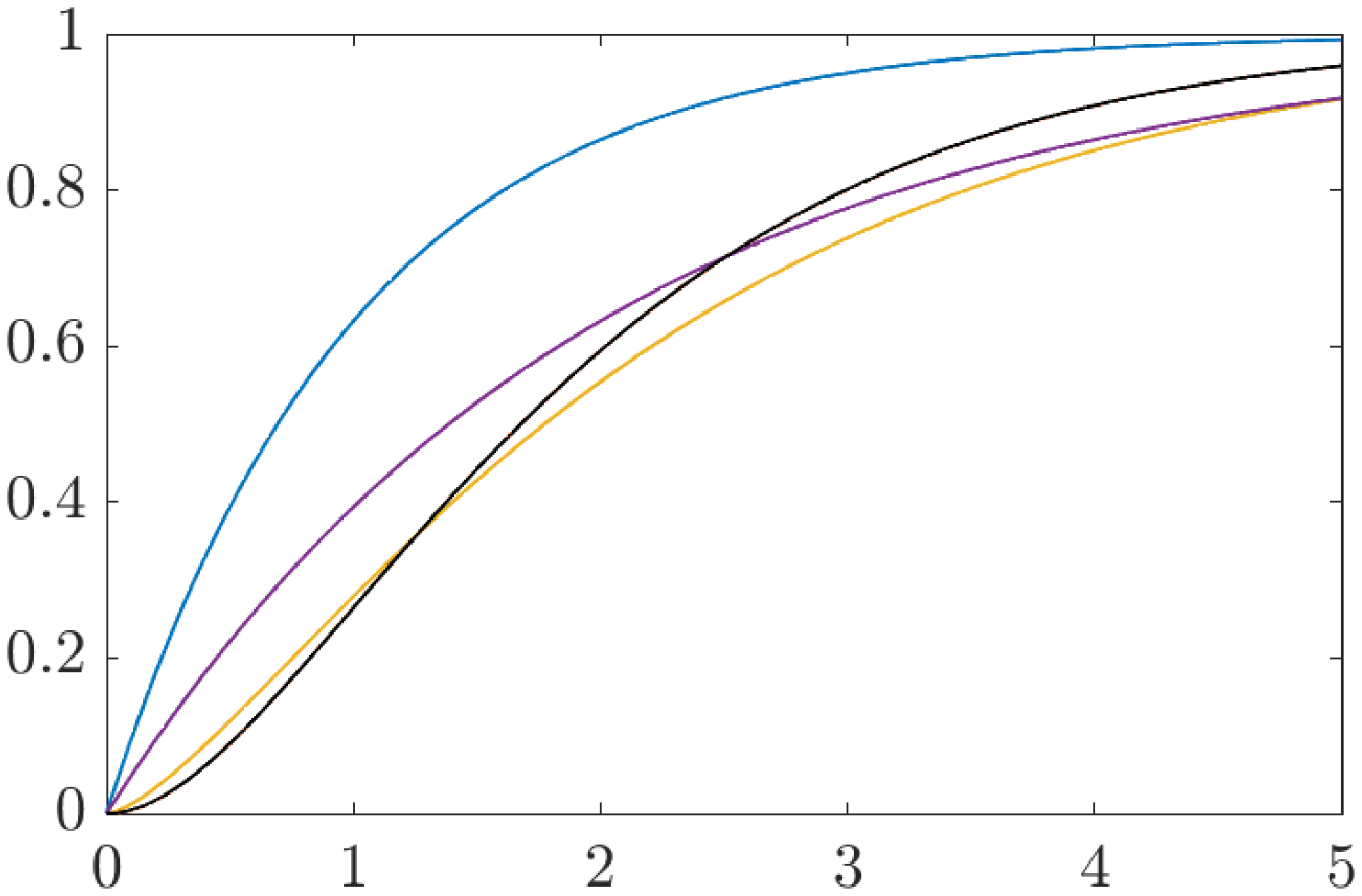}
\end{minipage}
\begin{minipage}[c]{0.46\linewidth}
\small
\begin{adjustbox}{width=1.1\textwidth, center=\textwidth}
\begin{tabular}{ l | *{4}{C{40pt}} } 
\hline \hline
& \multicolumn{4}{c}{Forecasting density} \\
& $\mathcal{G}a(1,1)$& $\mathcal{G}a(2,1)$& $\mathcal{G}a(\frac{3}{2},\frac{3}{2})$& $\mathcal{G}a(1,2)$\\ 
\hline
CRPS             & 3 & \textbf{1} & 4 & 2  \\ 
\AS$(\cdot,\cdot;0.05)$  & 4 & \textbf{1} & 2 & 3  \\ 
\AS$(\cdot,\cdot;0.275)$ & 4 & \textbf{1} & 2 & 3  \\ 
\AS$(\cdot,\cdot;0.5)$   & 4 & \textbf{1} & 2 & 3  \\ 
\AS$(\cdot,\cdot;0.725)$ & 4 & \textbf{1} & 2 & 3  \\ 
\AS$(\cdot,\cdot;0.95)$  & 2 & \textbf{1} & 3 & 4  \\ 
\hline 
\end{tabular}
\end{adjustbox}
\end{minipage}
\caption{Ranking of probabilistic forecasts. Results from $S=1$ simulation of $N=100$ observations. Density estimated with $M=500$ draws from forecasting distribution. Target is $\mathcal{G}a(2,1)$ (black), forecasting densities are: $\mathcal{G}a(1,1)$ (blue), $\mathcal{G}a(2,1)$ (orange), $\mathcal{G}a(\frac{3}{2},\frac{3}{2})$ (yellow), $\mathcal{G}a(1,2)$ (purple).}
\label{tab:test_GamGam}
\end{figure}

Fig.~\ref{tab:test_GamGam} presents the results for a $\mathcal{G}a(2,1)$ target density. By looking at the worst performing densities according to \AS, we find that $\mathcal{G}a(1,1)$ is assigned the highest penalty for values $c \leq 0.725$, while $\mathcal{G}a(1,2)$ becomes the worst for $c = 0.95$.
This reflects that for $c \leq 0.725$, the asymmetric score penalizes more the underestimation, while for $c = 0.95$ it gives more weight to overestimation.
Similar results are found in Fig.~\ref{tab:test_BeBe} with a positively skewed Beta target density, $\mathcal{B}e(1,2)$.

%

\begin{figure}[!th]
\centering
\renewcommand*{\arraystretch}{0.9}
\begin{minipage}[c]{0.46\linewidth}
\includegraphics[trim= 10mm 0mm 10mm 0mm,clip,height= 4.0cm, width= 6.0cm]{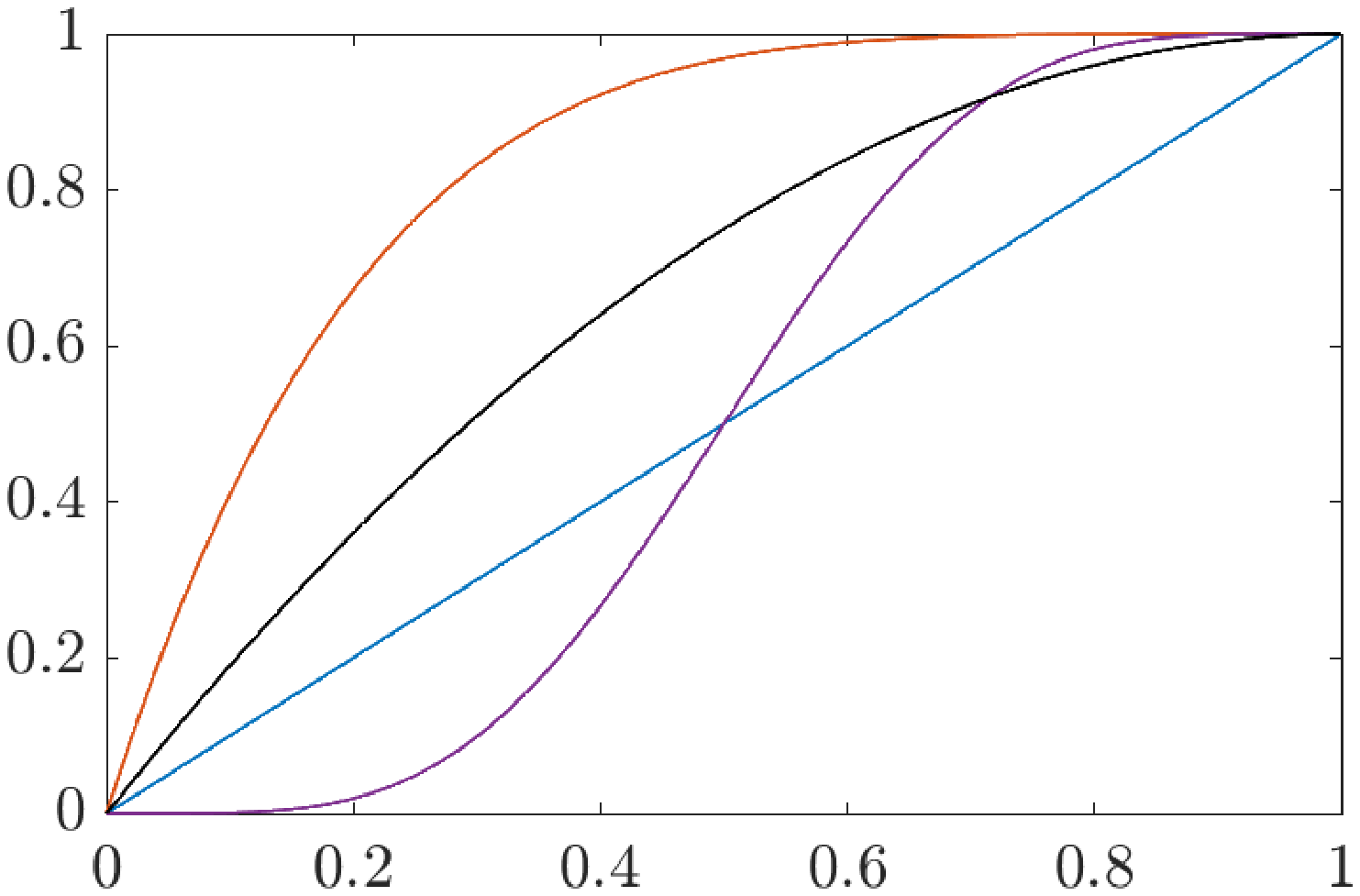}
\end{minipage}
\begin{minipage}[c]{0.46\linewidth}
\small
\begin{adjustbox}{width=1.1\textwidth, center=\textwidth}
\begin{tabular}{ l | *{4}{C{36pt}} } 
\hline \hline
& \multicolumn{4}{c}{Forecasting density} \\
& $\mathcal{B}e(1,1)$& $\mathcal{B}e(1,5)$& $\mathcal{B}e(1,2)$& $\mathcal{B}e(5,5)$\\ 
\hline
CRPS             & 3 & 2 & \textbf{1} & 4  \\ 
\AS$(\cdot,\cdot;0.05)$  & 2 & 3 & \textbf{1} & 4  \\ 
\AS$(\cdot,\cdot;0.275)$ & 3 & 2 & \textbf{1} & 4  \\ 
\AS$(\cdot,\cdot;0.5)$   & 2 & 3 & \textbf{1} & 4  \\ 
\AS$(\cdot,\cdot;0.725)$ & 3 & 4 & \textbf{1} & 2  \\ 
\AS$(\cdot,\cdot;0.95)$  & 2 & 4 & \textbf{1} & 3  \\
\hline 
\end{tabular}
\end{adjustbox}
\end{minipage}
\caption{Ranking of probabilistic forecasts. Results from $S=1$ simulation of $N=100$ observations. Density estimated with $M=500$ draws from forecasting distribution. Target is $\mathcal{B}e(1,2)$ (black), forecasting densities are: $\mathcal{B}e(1,1)$ (blue), $\mathcal{B}e(1,5)$ (orange), $\mathcal{B}e(1,2)$ (yellow), $\mathcal{B}e(5,5)$ (purple).}
\label{tab:test_BeBe}
\end{figure}

\subsection{Threshold-weighted version}
We deep further the properties of the proposed asymmetric scoring rule by considering a threshold-weighted version and comparing it with the threshold-weighted CRPS.
The goal is to disentangle the different role of the domain-weighting scheme, which reflects the interest of the decision-maker in having good forecasts within a specific interval of values, and of the error-weighting scheme, which corresponds to the decision-maker's loss in case of under or overestimation.

Consider a simulated experiment where $N=100$ observations are drawn from a Normal distribution $\mathcal{N}(1,4)$ and several forecasting densities are approximated using $M=500$ draws.
We consider the domain-weighting schemes in Tab.~\ref{tab:weights}, using $5$ alternative asymmetry levels $c \in\{ 0.05, 0.275, 0.50, 0.725, 0.95 \}$.

In Tab.~\ref{tab:threshold_weight_NormNorm} we find that the asymmetric penalty imposed by \AS \, plays a significant role for all domain-weighting schemes considered.
For an uniform weight, the \AS \, agrees with the CRPS for $c=0.5$, i.e. the symmetric case, but rewards differently the density forecasts for alternative values of the asymmetry level $c$.
When the interest is focused on the right tail of the distribution, both threshold-weighted CRPS and \AS \, agree, but when the attention is on the left tail, the two scoring rules perform remarkably different. The CRPS favours the standard Normal over the $\mathcal{N}(3,1)$, while the \AS \, rewards the latter for all $c \geq 0.275$.

The key insight obtained from this simulated exercise concerns the importance of domain- and error-weighting schemes.
The first assigns an heterogeneous weight to the performance on different intervals, while the  latter asymmetrically rewards negative and positive deviations from the true value.
The threshold-weighted asymmetric scoring rule, t\AS, combines the two schemes and allows to evaluate the performance of the forecasting density from both perspectives.
This is important to the decision makers, who are usually interested in a specific range of all possible values, thus calling for heterogeneous domain-weighting, and have asymmetric preferences towards under or overestimation, which motivates an asymmetric score.

\begin{table}[!th]
\centering
\vspace{-0.01in}
\renewcommand*{\arraystretch}{0.9}
\caption{This table reports the ranking of probabilistic forecasts using $t$CRPS and $t$\AS, for different weights (uniform, center, tails, right and left tail) and asymmetry levels ($c\  \in\{ 0.05, 0.275 0.50, 0.725, 0.95 \}$).
Results from $S=1$ simulation of $N=100$ observations (average score across all observations). Density estimated with $M=500$ draws from forecasting distribution. Target is $\mathcal{N}(1,4)$, forecasting densities are $\mathcal{N}(0,1)$, $\mathcal{N}(-3,1)$, $\mathcal{N}(3,1)$, $\mathcal{N}(0,16)$.}
\vspace*{2ex}
\begin{adjustbox}{width=0.8\textwidth, center=\textwidth}
\begin{tabular}{ l | *{4}{c} } 
\hline \hline
& $\mathcal{N}(0,1)$& $\mathcal{N}(-3,1)$& $\mathcal{N}(3,1)$& $\mathcal{N}(0,16)$\\ 
\hline
tCRPS uniform                & 4 & 2 & 3 & \textbf{1}  \\ 
t\AS$(\cdot,\cdot;0.05)$ uniform     & \textbf{1} & 3 & 4 & 2  \\ 
t\AS$(\cdot,\cdot;0.275)$ uniform    & 2 & \textbf{1} & 4 & 3  \\ 
t\AS$(\cdot,\cdot;0.5)$ uniform      & 4 & 2 & 3 & \textbf{1}  \\ 
t\AS$(\cdot,\cdot;0.725)$ uniform    & 4 & 3 & 2 & \textbf{1}  \\ 
t\AS$(\cdot,\cdot;0.95)$ uniform     & 4 & 3 & \textbf{1} & 2  \\ 
\hline
tCRPS center                 & \textbf{1} & 3 & 4 & 2  \\ 
t\AS$(\cdot,\cdot;0.05)$ center      & 3 & \textbf{1} & 4 & 2  \\ 
t\AS$(\cdot,\cdot;0.275)$ center     & 3 & \textbf{1} & 4 & 2  \\ 
t\AS$(\cdot,\cdot;0.5)$ center       & 4 & \textbf{1} & 3 & 2  \\ 
t\AS$(\cdot,\cdot;0.725)$ center     & 4 & \textbf{1} & 3 & 2  \\ 
t\AS$(\cdot,\cdot;0.95)$ center      & 4 & \textbf{1} & 3 & 2  \\ 
\hline
tCRPS tails                  & \textbf{1} & 3 & 4 & 2  \\ 
t\AS$(\cdot,\cdot;0.05)$ tails       & \textbf{1} & 4 & 2 & 3  \\ 
t\AS$(\cdot,\cdot;0.275)$ tails      & \textbf{1} & 3 & 2 & 4  \\ 
t\AS$(\cdot,\cdot;0.5)$ tails        & 2 & 3 & 4 & \textbf{1}  \\ 
t\AS$(\cdot,\cdot;0.725)$ tails      & 3 & 4 & 2 & \textbf{1}  \\ 
t\AS$(\cdot,\cdot;0.95)$ tails       & 4 & 3 & \textbf{1} & 2  \\ 
\hline
tCRPS right tail             & 2 & 3 & 4 & \textbf{1}  \\ 
t\AS$(\cdot,\cdot;0.05)$ right tail  & 3 & 2 & 4 & \textbf{1}  \\ 
t\AS$(\cdot,\cdot;0.275)$ right tail & 3 & 2 & 4 & \textbf{1}  \\ 
t\AS$(\cdot,\cdot;0.5)$ right tail   & 4 & 2 & 3 & \textbf{1}  \\ 
t\AS$(\cdot,\cdot;0.725)$ right tail & 4 & 3 & 2 & \textbf{1}  \\ 
t\AS$(\cdot,\cdot;0.95)$ right tail  & 4 & 3 & \textbf{1} & 2  \\ 
\hline
tCRPS left tail              & \textbf{1} & 3 & 2 & 4  \\
t\AS$(\cdot,\cdot;0.05)$ left tail   & \textbf{1} & 3 & 4 & 2  \\ 
t\AS$(\cdot,\cdot;0.275)$ left tail  & 2 & 3 & \textbf{1} & 4  \\ 
t\AS$(\cdot,\cdot;0.5)$ left tail    & 4 & 3 & \textbf{1} & 2  \\ 
t\AS$(\cdot,\cdot;0.725)$ left tail  & 4 & 3 & \textbf{1} & 2  \\ 
t\AS$(\cdot,\cdot;0.95)$ left tail   & 4 & 2 & \textbf{1} & 3  \\ 
\hline 
\end{tabular}
\end{adjustbox}
\label{tab:threshold_weight_NormNorm} 
\end{table}

\section{Empirical applications}    \label{sec:application}
In the empirical applications, we adopt a similar framework to \cite{amisano2007comparing} and \cite{gneiting2011comparing}, and consider the task of comparing density forecasts in a time series context. We use a fixed-length rolling window to provide a density forecast for $h$ step ahead future observations. We focus on three different applications related to macroeconomics (employment growth rate) and to commodity prices (oil prices and electricity prices). We compare several univariate models, such as the autoregressive (AR) model, the Markov-switching (MS) AR model and the time-varying parameter (TVP) AR model.

We use the AR(1) as benchmark model, then we specify $12$ lags for the employment growth rate (i.e., $1$ year of monthly observations) and $20$ lags for the oil (i.e., $1$ month of daily observations). Regarding the electricity prices, we include $7$ lags (i.e., $1$  week of daily observations) and by following common practice in the literature, we restrict lags to $t-1$, $t-2$, and $t-7$, which correspond to the previous day, two days before, and one week before the delivery time, recalling first similar conditions that may have characterized the market over the same hours and similar days (such as congestions and blackouts) and secondly the demand level during the days of the week.
For the MS-AR model we consider only $1$ lag, while for the TVP-AR model we use $1$ and $2$ lags.
For both AR and TVP-AR, we consider three specifications of the variance: constant volatility and time-varying volatility in the form of stochastic volatility with Gaussian and Student-t error. For the MS-AR, we impose an identification constraint on the error variance.

In the first application, we aim at forecasting monthly US total nonfarm seasonally adjusted employment growth rate downloaded from the FRED database. We consider the growth rate of the monthly employment rate in US from January 1980 to April 2020. We see evidence of some spikes, in particular with a strong fall in April 2020 due to present COVID-19 situation (see Figure~\ref{fig:data_rraw} in the supplementary material). We use a rolling window approach of $20$ years (thus $240$ observations) and we forecast $h=1$ and $h=12$ (thus $1$ year ahead) month ahead by using a recursive forecasting exercise. 

For oil prices, we analyze daily West Texas Index (WTI) data (no weekends) from 02 January 2012 to 07 May 2020 in order to include in the analysis the recent turmoil. Indeed, large drops in demand that suddenly occurred and storage scarcity have resulted in negative WTI oil prices at the end of April 2020. As for the employment rate, we have used a rolling window of $4$ years and we forecast $h=1$ and $h=5$ days ahead by applying recursive techniques.

In the third application, we consider the problem of forecasting the day-ahead electricity prices in Germany, one of the largest and leading energy market. In the electricity markets, the phenomenon of negative prices -- when allowed to occur, such as in Germany where there is no floor price -- has become more frequent due to the increasing share of electricity generated from renewable energy sources (RES) and the current impossibility to store it (see Figure 2 in the supplementary material). We analyze daily data (with weekends) from 01 January 2014 to 08 May 2020. For the forecasting analysis, we have considered a rolling window of $3$ years and a recursive techniques for predicting $h=1$ and $h=7$ days ahead.

As we discussed in the introduction, policymakers or energy producers may be more concerned with forecasting values below a given threshold than the full distribution, since they require different measures, including in the case of energy variables to stop the production.\footnote{Unfortunately, we have not precise data to compute (i) the value of this threshold, excluding the case of RES producers of electricity prices, that could be still profitable even when prices are marginally above zero, and (ii) the level of asymmetry of the loss function. Therefore, we investigate several values of $c$, the parameter that drives the asymmetry of our measure.} This supports the application of the \AS. For the oil series we perform a case study around the collapse of WTI prices and discuss how the \AS \, results can be applied to identify the true unknown density.

Before evaluating the relative performance of all models, we check the calibration of the density forecasts. Calibration of density forecasts is based on properties of a density and refers to absolute accuracy (see \cite{BCR2019} for further details). The absolute accuracy can be studied by testing forecast accuracy relative to the ``true'', unobserved density.
\cite{Dawid82} introduced the criterion of calibration for comparing prequential probabilities with binary random outcomes and exploited the concept of probability integral transform (PIT), that is the value that a predictive CDF attains at the observations, for continuous random variables. The PITs summarize the properties of the densities and may help us to judge whether the densities are biased in a particular direction and whether the width of the densities has been roughly correct on average, see \cite{DieboldGunterTay1998}. The PITs can provide an indication of whether a density is wrong in predicting higher moments or specific parts of the distribution, such as the tails; however they cannot distinguish among models that are also correctly calibrated. We apply the test of \cite{Knuppel2015} and refer to \cite{RossiSekhposyan2013} for evaluation of PITs in presence of instabilities, \cite{RossiSekhposyan2014} for application with large database and \cite{RossiSekhposyan2016} for a comparison of alternative tests for correct specification of density forecasts.

The PIT tests in Tab.~\ref{tab:Real_CRPS_ACPS} indicate that all densities are correctly calibrated for the employment growth rate at 5\% significance level, excluding the one given by the TVP-AR(2) model at 12-month horizon, for which the p-value is marginally lower at 4.9\%. Density forecasts from models TVP-AR(2)-SV\footnote{Notice that the TVP-AR(2)-SV is always preferred in terms of relative accuracy.} and TVP-AR(2)-tSV are calibrated at 1-day ahead horizon; no density is correctly calibrated at 5-days ahead horizons. All densities are not correctly calibrated when predicting EEX electricity prices at both horizons. So, the PITs analysis suggests there is not a stochastically dominating model, but more specifications can provide (absolute) accurate forecasts suggesting the use of relative metrics such as the \AS \, to discriminate among them. In the case of EEX prices, all models are wrong and a possible explanation is that the models considered in this text are based only on econometric properties of the series, hence they may be labelled as ``purely econometric'' models. \cite{GRR2020a} and \cite{GRR2020b} document how important is to extend these models with economically relevant variables, such as variables related to the demand and the production of electricity, including renewable energy sources, to increase accuracy.
We leave this extension for further research and apply our metrics to an example where models in terms of calibration are all wrong.   

\newcommand{\unastar}{$^{\ast}$}
\newcommand{\duestar}{$^{\ast \ast}$}
\newcommand{\trestar}{$^{\ast \ast \ast}$}
\newcommand{\cg}{\cellcolor[gray]{0.85}}

\begin{table}[!th] 
\centering 
\caption{Ranking of probability forecasts and accuracy test. Best model, over vintages, according to: CRPS; ACPS with $c=0.05; 0.5; 0.95$ for the three different datasets: Employment (top); Oil (middle) and EEX (bottom).} 
\begin{adjustbox}{width=1.3\textwidth, center=\textwidth}
\begin{threeparttable}
\begin{tabular}{l*{13}{c}} 
\toprule 
& & & &  & & & \multicolumn{2}{c}{\textsc{EMPL}} & \\
\midrule 
\multicolumn{1}{l}{\textit{Horizon 1}} & \cg AR(1)& \cg AR(1)-SV& \cg  AR(1)-tSV& \cg AR(12)& \cg AR(12)-SV& \cg AR(12)-tSV& \cg AR(1)-MS& \cg TVP-AR(1)& \cg TVP-AR(1)-SV& \cg TVP-AR(1)-tSV& \cg TVP-AR(2)& \cg TVP-AR(2)-SV& \cg TVP-AR(2)-tSV \\ 
\AS$(\cdot,\cdot;0.05)$ & 
12 & 9\trestar & 5\trestar & 13 & 2\trestar & 1\trestar & 11\trestar & 3\trestar & 10\trestar & 8\trestar & 4\trestar & 7\trestar & 6\trestar \\
\AS$(\cdot,\cdot;0.5)$ & 13 & 10\trestar & 9\trestar & 11 & 8\trestar & 7\trestar & 12\trestar & 4\trestar & 6\trestar & 5\trestar & 1\trestar & 2\trestar & 3\trestar \\
\AS$(\cdot,\cdot;0.95)$   & 13 & 4\trestar & 3\trestar & 12 & 1\trestar & 2\trestar & 11\trestar & 5\trestar & 7\trestar & 9\trestar & 6\trestar & 8\trestar & 10\trestar \\
CRPS & 13 & 10\trestar & 9\trestar & 11 & 8\trestar & 7\trestar & 12\trestar & 4\trestar & 6\trestar & 5\trestar & 1\trestar & 2\trestar & 3\trestar \\
\midrule
\multicolumn{1}{l}{\textit{Horizon 12}}  & \cg AR(1)& \cg AR(1)-SV& \cg AR(1)-tSV& \cg AR(12)& \cg AR(12)-SV& \cg AR(12)-tSV& \cg AR(1)-MS& \cg TVP-AR(1)& \cg TVP-AR(1)-SV& \cg TVP-AR(1)-tSV& TVP-AR(2)& \cg TVP-AR(2)-SV& \cg TVP-AR(2)-tSV\\ 
\AS$(\cdot,\cdot;0.05)$ & 11 & 12 & 10 & 13 & 3 & 2 & 9\trestar & 4 & 6 & 5 & 1 & 8 & 7 \\
\AS$(\cdot,\cdot;0.5)$ & 12 & 4\trestar & 3\trestar & 13 & 1\trestar & 2\trestar & 11\trestar & 5\trestar & 7\trestar & 8\trestar & 6\trestar & 9\trestar & 10\trestar \\
\AS$(\cdot,\cdot;0.95)$   & 12 & 1\trestar & 2\trestar & 13 & 3\trestar & 4\trestar & 11\trestar & 5\trestar & 7\trestar & 8\trestar & 6\trestar & 9\trestar & 10\trestar \\
CRPS & 12 & 4\trestar & 3\trestar & 13 & 1\trestar & 2\trestar & 11\trestar & 5\trestar & 7\trestar & 8\trestar & 6\trestar & 9\trestar & 10\trestar \\
\bottomrule 
\toprule 
& & & &  & & & \multicolumn{2}{c}{\textsc{OIL}} & \\
\midrule 
\multicolumn{1}{l}{\textit{Horizon 1}}  & AR(1)& AR(1)-SV& AR(1)-tSV& AR(20)& AR(20)-SV& AR(20)-tSV& AR(1)-MS& TVP-AR(1)& TVP-AR(1)-SV& TVP-AR(1)-tSV& TVP-AR(2)& \cg TVP-AR(2)-SV& \cg TVP-AR(2)-tSV \\  
\AS$(\cdot,\cdot;0.05)$ & 12 & 8\unastar & 11 & 10\duestar & 7\unastar & 9\duestar & 13 & 6\trestar & 3\trestar & 4\trestar & 5\trestar & 2\trestar & 1\trestar \\
\AS$(\cdot,\cdot;0.5)$ & 12 & 9 & 13 & 7 & 8 & 11 & 10 & 3\trestar & 6\trestar & 5\trestar & 1\trestar & 4\trestar & 2\trestar \\
\AS$(\cdot,\cdot;0.95)$   & 12 & 8 & 11\duestar & 9 & 7 & 10\duestar & 13 & 3\duestar & 6\unastar & 5\unastar & 1\duestar & 4\unastar & 2\unastar \\
CRPS & 12 & 9 & 13 & 7 & 8 & 11 & 10 & 3\trestar & 6\trestar & 5\trestar & 1\trestar & 4\trestar & 2\trestar \\
\midrule
\multicolumn{1}{l}{\textit{Horizon 5}} & AR(1)& AR(1)-SV& AR(1)-tSV& AR(20)& AR(20)-SV& AR(20)-tSV& AR(1)-MS& TVP-AR(1)& TVP-AR(1)-SV& TVP-AR(1)-tSV& TVP-AR(2)& TVP-AR(2)-SV& TVP-AR(2)-tSV\\  
\AS$(\cdot,\cdot;0.05)$ & 10 & 5 & 7 & 9 & 3 & 6 & 12 & 13 & 2 & 8 & 11 & 1\unastar & 4\unastar \\
\AS$(\cdot,\cdot;0.5)$ & 3 & 7 & 6 & 8 & 4 & 5 & 11 & 1\duestar & 13 & 12 & 2 & 10 & 9 \\
\AS$(\cdot,\cdot;0.95)$   & 6 & 1 & 3\unastar  & 8 & 2 & 4 & 13 & 5 & 10 & 7 & 11 & 12 & 9 \\
CRPS & 3 & 6 & 7 & 8 & 4 & 5 & 11 & 1\duestar & 13 & 12 & 2 & 10 & 9 \\
\bottomrule 
\toprule 
& & & &  & & & \multicolumn{2}{c}{\textsc{EEX}} & \\
\midrule 
\multicolumn{1}{l}{\textit{Horizon 1}}  & AR(1)& AR(1)-SV& AR(1)-tSV& AR(20)& AR(20)-SV& AR(20)-tSV& AR(1)-MS& TVP-AR(1)& TVP-AR(1)-SV& TVP-AR(1)-tSV& TVP-AR(2)& TVP-AR(2)-SV& TVP-AR(2)-tSV\\ 
\AS$(\cdot,\cdot;0.05)$ &11 & 8\unastar & 6\trestar & 9\trestar & 3\trestar & 1\trestar & 12 & 13 & 7\duestar & 4\duestar & 10 & 5\duestar & 2\trestar \\
\AS$(\cdot,\cdot;0.5)$ & 12 & 10\trestar & 11\unastar & 7\trestar & 2\trestar & 5\trestar & 13 & 9\trestar & 6\trestar & 4\trestar & 8\trestar & 3\trestar & 1\trestar \\
\AS$(\cdot,\cdot;0.95)$   & 12 & 7\trestar & 8\trestar & 11\trestar & 2\trestar & 6\trestar & 13 & 9\trestar & 3\trestar & 5\trestar & 10\trestar & 1\trestar & 4\trestar \\
CRPS & 13 & 10\trestar & 11\unastar & 7\trestar & 1\trestar & 5\trestar & 12 & 9\trestar & 6\trestar & 4\trestar &  8\trestar & 3\trestar & 2\trestar \\
\midrule
\multicolumn{1}{l}{\textit{Horizon 7}}  & AR(1)& AR(1)-SV& AR(1)-tSV& AR(20)& AR(20)-SV& AR(20)-tSV& AR(1)-MS& TVP-AR(1)& TVP-AR(1)-SV& TVP-AR(1)-tSV& TVP-AR(2)& TVP-AR(2)-SV& TVP-AR(2)-tSV\\  
\AS$(\cdot,\cdot;0.05)$ & 5 & 13 & 12 & 1\trestar & 11 & 4 & 7 & 3 & 9 & 8 & 2 & 10 & 6 \\
\AS$(\cdot,\cdot;0.5)$ & 10 & 12 & 13 & 9\trestar & 8\trestar & 7\trestar & 11 & 6\trestar & 4\trestar & 2\trestar & 5\trestar & 3\trestar & 1\trestar \\
\AS$(\cdot,\cdot;0.95)$   & 12 & 11 & 10 & 9\trestar & 7\unastar & 5\duestar & 13 & 1\trestar & 8\unastar & 4\duestar & 2\trestar & 6\unastar & 3\duestar \\
CRPS & 10 & 12 & 13 & 9\trestar & 8\trestar & 7\trestar & 11 & 6\trestar & 4\trestar & 2\trestar & 5\trestar & 3\trestar & 1\trestar \\
\bottomrule 
\end{tabular}
\small{
\begin{tablenotes}
\item \textit{Notes:}
\item[1] $^{\ast \ast \ast}$, $^{\ast \ast}$ and $^{\ast}$ indicate scores are significantly different from 1 at $1\%$, $5\%$ and $10\%$, according to the \AS test in Section \ref{sec:test}.
\item[2] Gray cells indicate models that are correctly
calibrated at $5\%$ significance level according to the Knuppel test. 
\end{tablenotes}
}
\end{threeparttable}
\end{adjustbox}
\label{tab:Real_CRPS_ACPS}
\end{table} 

Tab.~\ref{tab:Real_CRPS_ACPS} shows the ranking of the probability forecasts over vintages and across models for all the three datasets for $c=0.05,\,0.5,\,0.95$.\footnote{See Table IV in the supplementary material for results for a higher range of $c$.} The \AS test presented in Section \ref{sec:test} is also reported.\footnote{In order to perform the test, we checked the stationarity and short memory of the loss differential series using the ADF test and the autocorrelation function, respectively.}

Regarding the employment growth rate, we can see at horizon 1-month ahead that the best model for $c=0.05$ is the AR(12)-tSV, for $c=0.5$ it is the TVP-AR with $2$ lags (the same for the CRPS measure), and for $c=0.95$ it is the AR(12)-SV, showing differences across different levels of asymmetry. The test indicates that most of the models provide superior forecasts than the AR(1) benchmark and only the AR(12) model does not provide gains. The difference in model performance for various levels of $c$ is confirmed for $h=12$ and interesting for $c=0.05$ only the AR(1)-MS is statistically superior. Therefore, our evidence supports the large literature on the use of time-varying and nonlinear models in modelling and forecasting (un)employment data. Moreover, the best model for $h=12$ and $c=0.5$ is the same when applying the CRPS.  In Fig.~\ref{fig:EMPL_ACPS}, we report the best model in each vintage for the two horizon ahead, where the black line refers to the CRPS, the red and the yellow for the \AS \, for $c=0.05$ and for $c=0.95$, respectively. The graph shows large instability in the best model, in particular when using the CRPS. The \AS \, rules seem to prefer one of the alternative models for more consecutive vintages. For example, by looking at the relative frequency of occurrence of each model as the best model, we find that for $c=0.05$, $31\%$ times the AR(12)-tSV is considered the best model for $h=1$. Similar percentages are found for other levels of $c$ and $h$, despite model order varies substantially across measures.

\begin{figure}[!th]
\centering
\setlength{\tabcolsep}{0.2pt}
\setlength{\abovecaptionskip}{-3pt}
\hspace*{-12ex}
\begin{tabular}{cc}
\begin{rotate}{90} {\footnotesize \hspace*{45pt} $h=1$} \end{rotate} &
\includegraphics[trim= 25mm 5mm 40mm 5mm,clip,height= 4.0cm, width= 17.5cm]{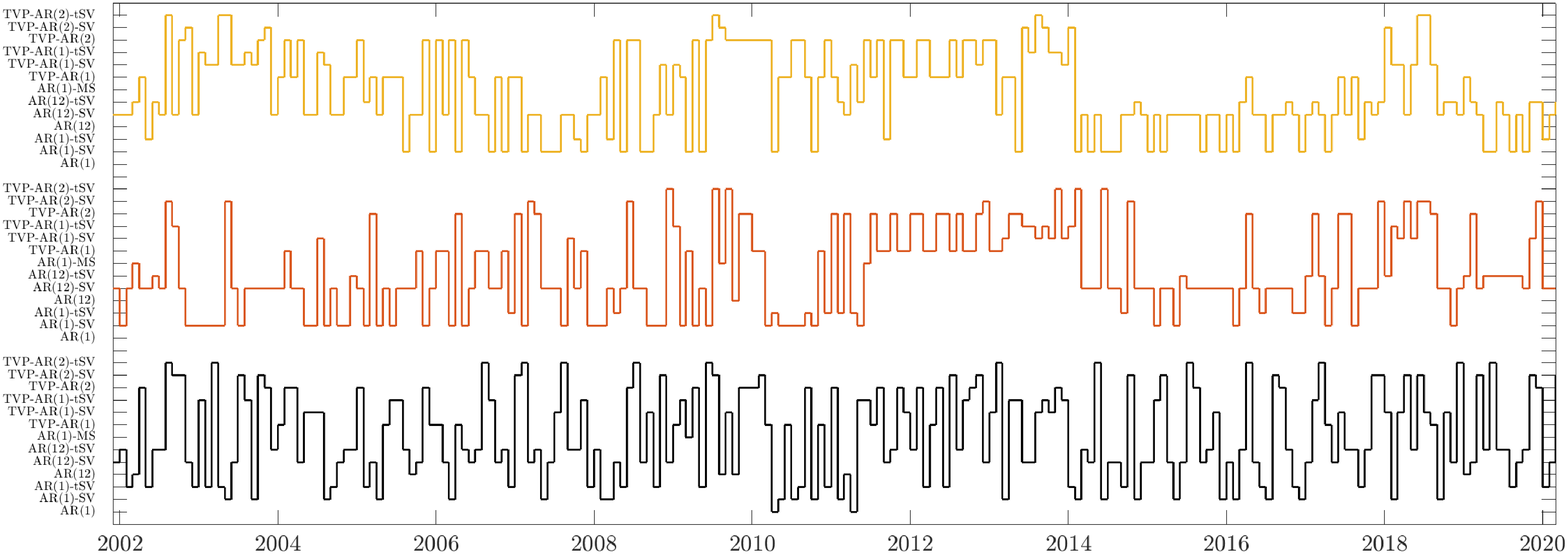} \\[-10pt]
\begin{rotate}{90} {\footnotesize \hspace*{45pt} $h=12$} \end{rotate} &
\includegraphics[trim= 25mm 5mm 40mm 5mm,clip,height= 4.0cm, width= 17.5cm]{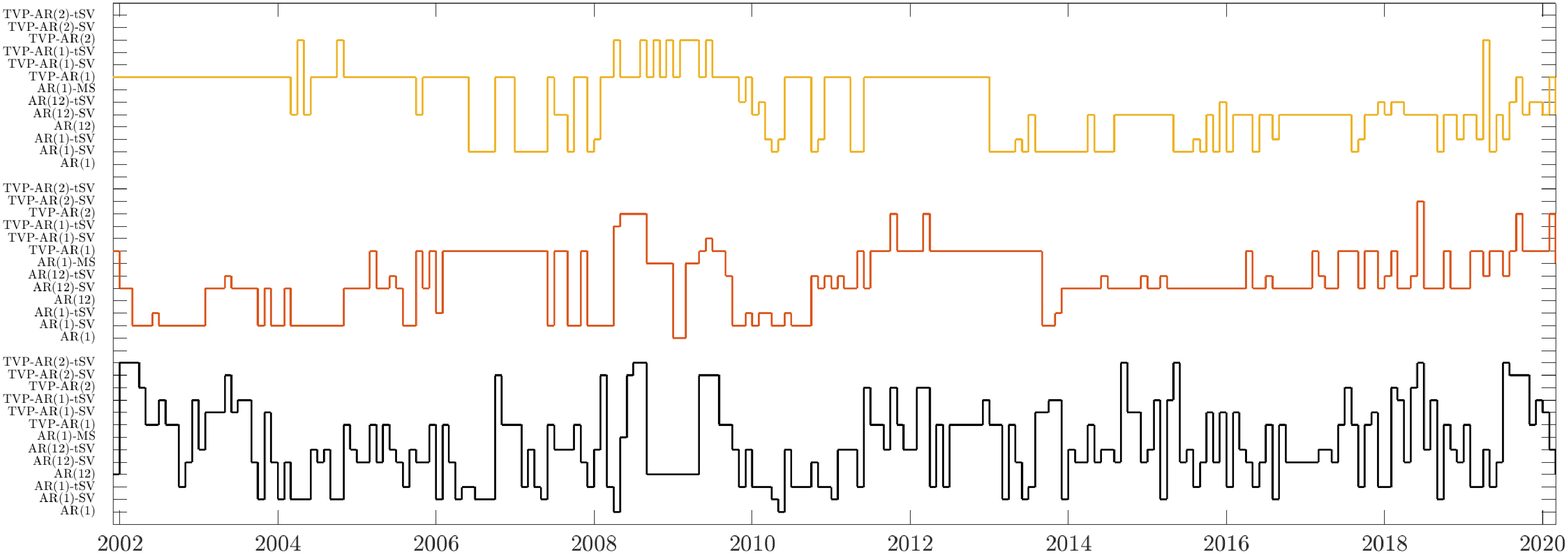}
\end{tabular}
\caption{Best model in each vintage for EMPL dataset: CRPS (black), \AS \, with $c=0.05$ (red), \AS \, with $c=0.95$ (yellow).}
\label{fig:EMPL_ACPS}
\end{figure}

Moving to the oil prices, in the middle panel of Tab.~\ref{tab:Real_CRPS_ACPS}, we find that across vintages, for $1$ day ahead the TVP-AR(2) is the best model whereas the TVP-AR(2)-SV is the second best for the asymmetric levels $c=0.5,\,0.95$ and the CRPS. For $c=0.05$ the best model is the TVP-AR(2)-SV model, supporting PITS evidence that this model is among the few ones correctly calibrated. The TVP-AR(2)-SV model is again the best model for $1$ week ahead of forecasting and for $c=0.05$ and it is one of the two models to be statistically superior to the AR benchmark. For the same weekly horizon and other levels of $c$, again only few models are superior to the benchmark. Fig.~\ref{fig:OIL_ACPS} confirms that the \AS \ is less variable in this selection than the CRPS.

\begin{figure}[!th]
\centering
\setlength{\tabcolsep}{0.2pt}
\setlength{\abovecaptionskip}{-3pt}
\hspace*{-12ex}
\begin{tabular}{cc}
\begin{rotate}{90} {\footnotesize \hspace*{45pt} $h=1$} \end{rotate} &
\includegraphics[trim= 25mm 5mm 40mm 5mm,clip,height= 4.0cm, width= 17.5cm]{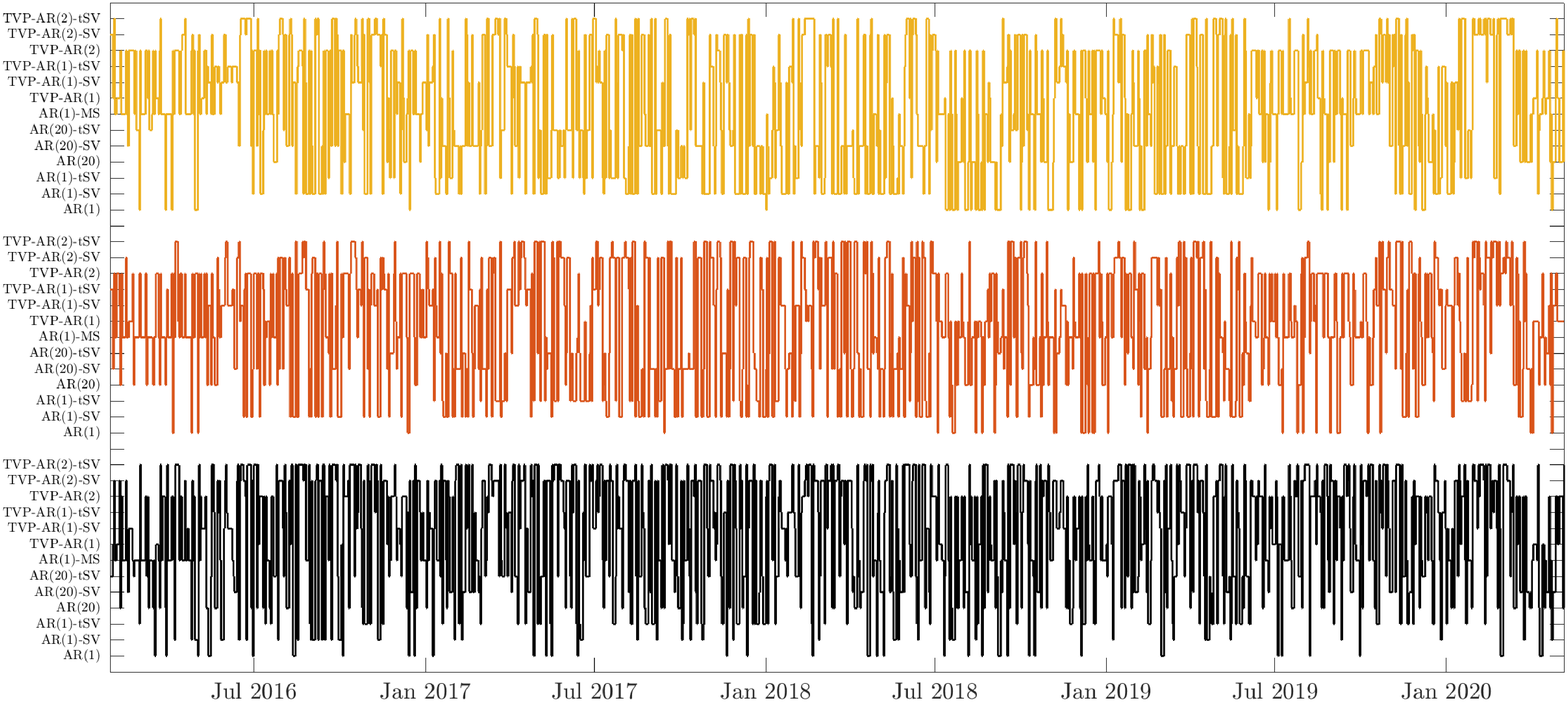} \\[-10pt]
\begin{rotate}{90} {\footnotesize \hspace*{45pt} $h=5$} \end{rotate} &
\includegraphics[trim= 25mm 5mm 40mm 5mm,clip,height= 4.0cm, width= 17.5cm]{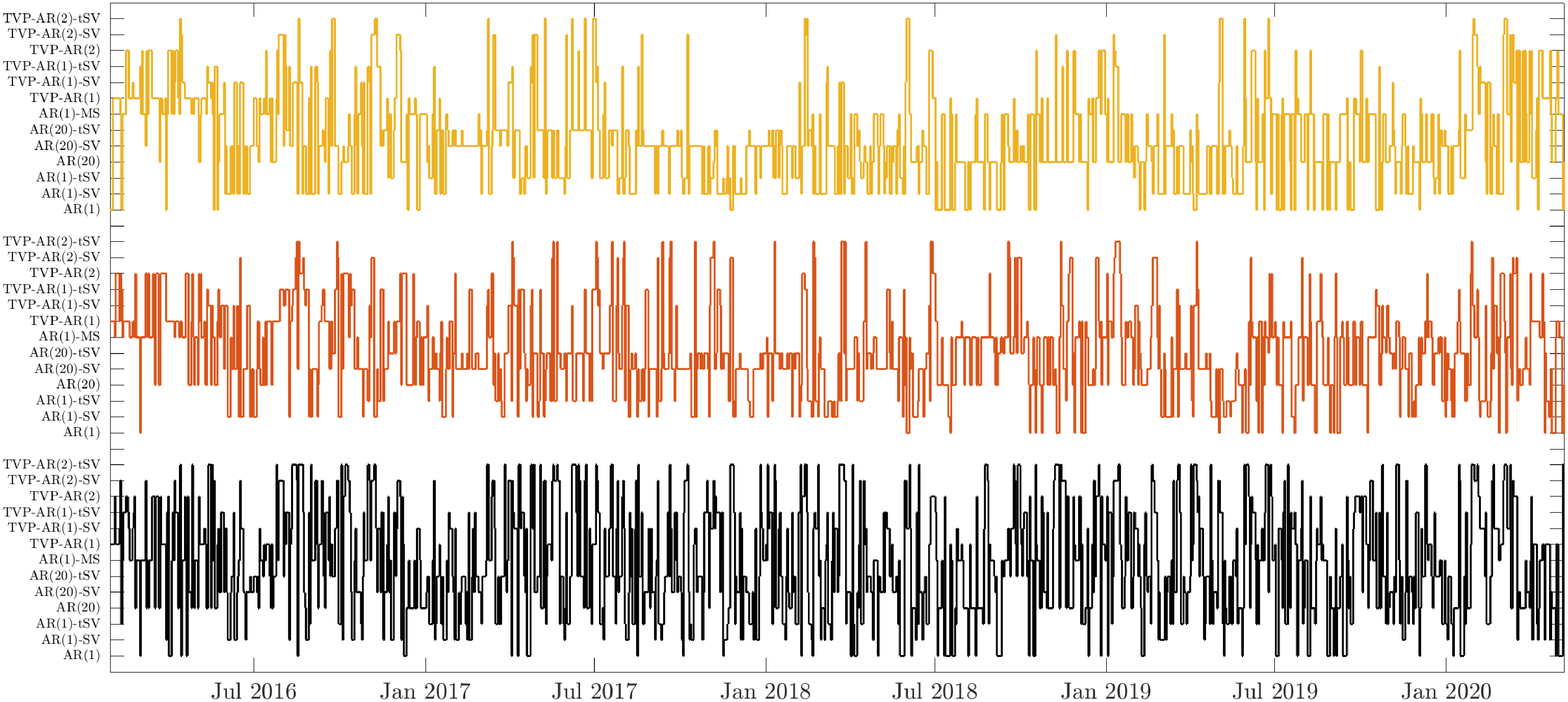}
\end{tabular}
\caption{Best model in each vintage for OIL dataset: CRPS (black), \AS \, with $c=0.05$ (red), \AS \, with $c=0.95$ (yellow).}
\label{fig:OIL_ACPS}
\end{figure}

Fig.~\ref{fig:example_score_OIL_AR20} illustrates the \AS \ for one step ahead density forecasts of the OIL prices, according to an AR(20) model, for each vintage of the rolling estimation and various levels of asymmetry.
Despite showing the results for a single model, this figure presents some interesting insights.
By looking at the scores between April 17 and April 21, we find that the forecast is worst performing for $c=0.05$ and best for $c=0.95$, indicating that the density forecast assigns more mass on the right part of the support as compared to the density of the observations. This situation is similar to the yellow line in Fig.~\ref{fig:score_varyingP}.
Surprisingly, the ranking is reversed between April 21 and April 24, where the forecast receives a higher score under $c=0.05$. This suggests that the density forecast is likely to be a right-shifted version of the observation density, similarly to the blue line in Fig.~\ref{fig:score_varyingP}.

\begin{figure}[!th]
\centering
\setlength{\abovecaptionskip}{0pt}
\hspace*{-6.5ex}
\begin{tabular}{ c c } 
\includegraphics[trim= 20mm 0mm 20mm 5mm,clip,height= 3.5cm, width= 7.5cm]{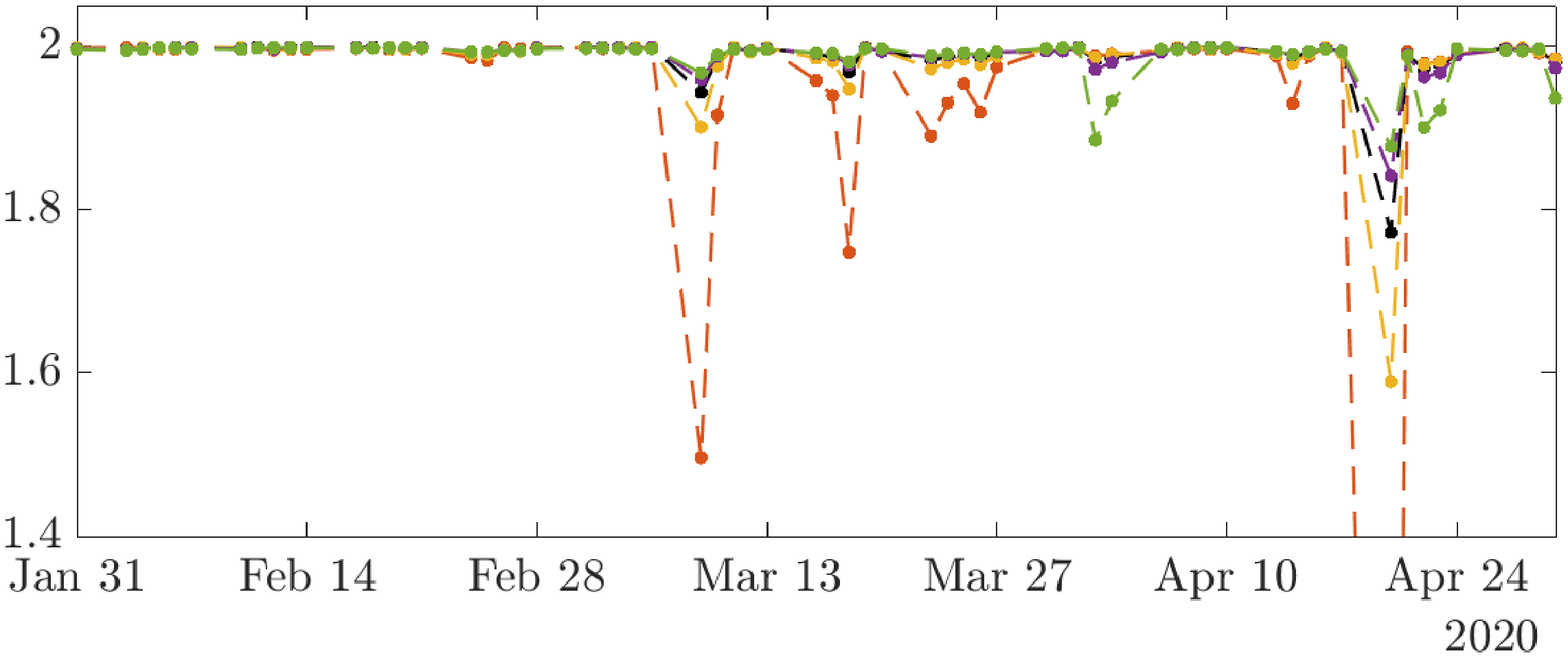} & 
\includegraphics[trim= 20mm 0mm 15mm 5mm,clip,height= 3.5cm, width= 7.5cm]{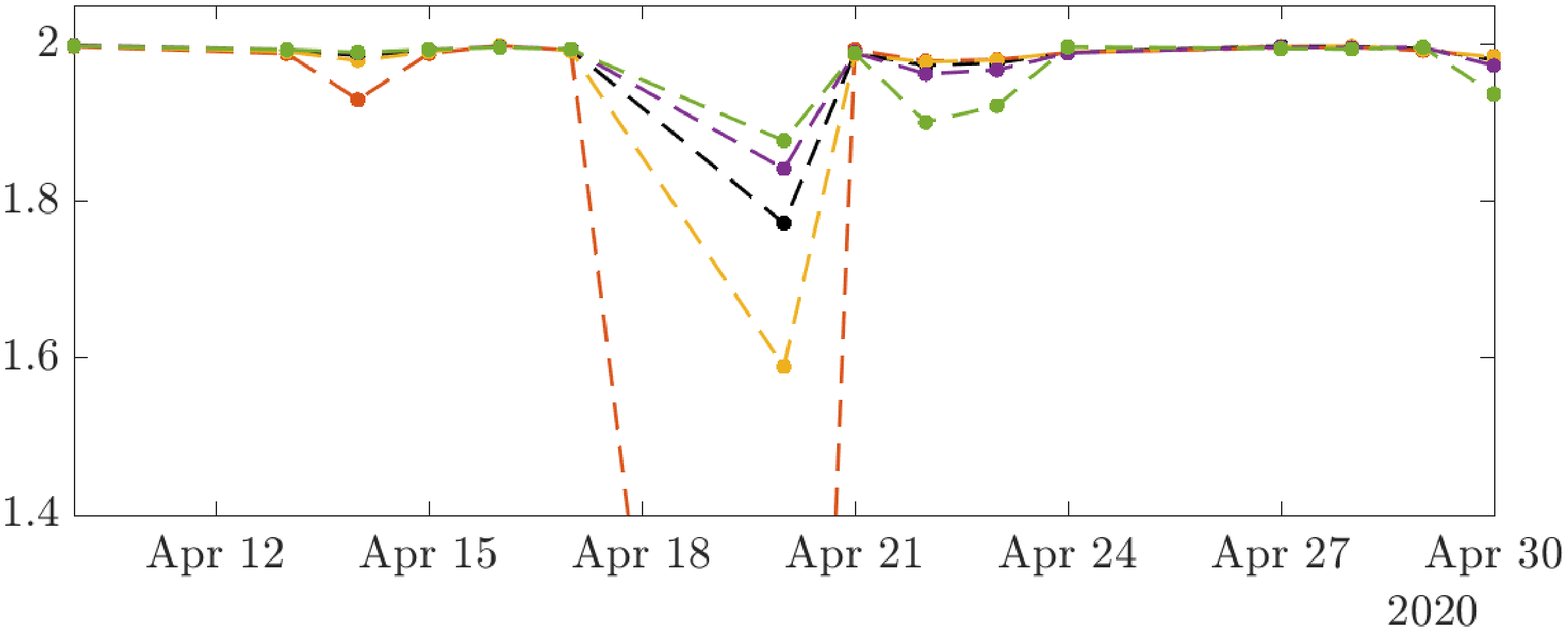} \\[-7pt]
\includegraphics[trim= 20mm 0mm 20mm 5mm,clip,height= 3.5cm, width= 7.5cm]{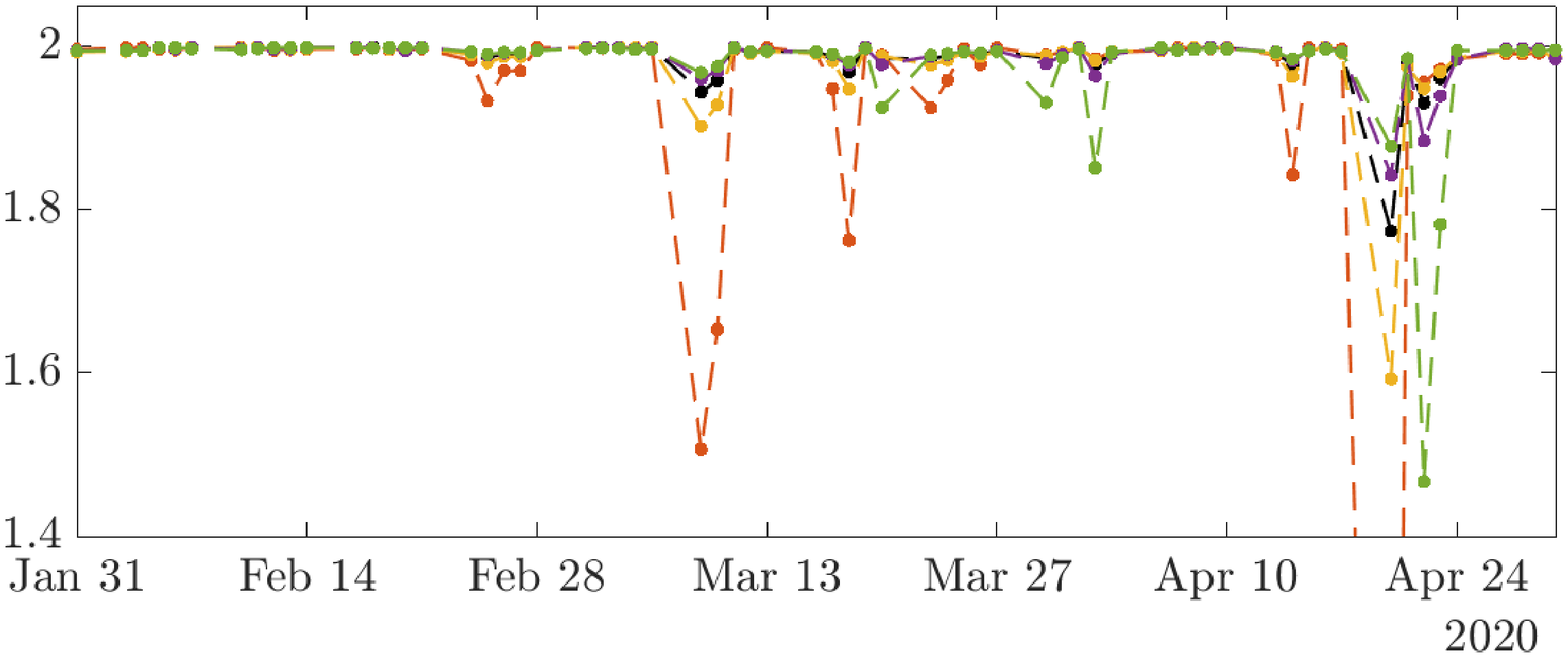} & 
\includegraphics[trim= 20mm 0mm 15mm 5mm,clip,height= 3.5cm, width= 7.5cm]{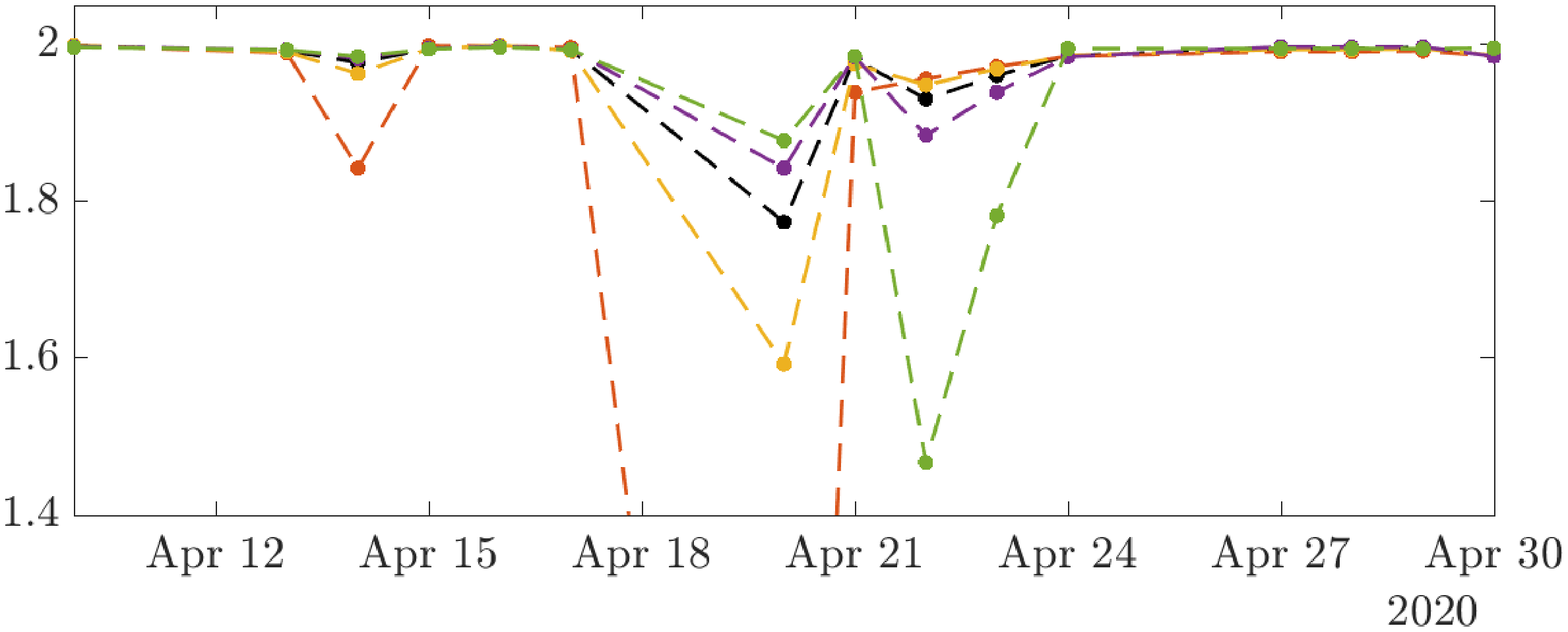} \\[-7pt]
\includegraphics[trim= 20mm 2mm 20mm 5mm,clip,height= 3.5cm, width= 7.5cm]{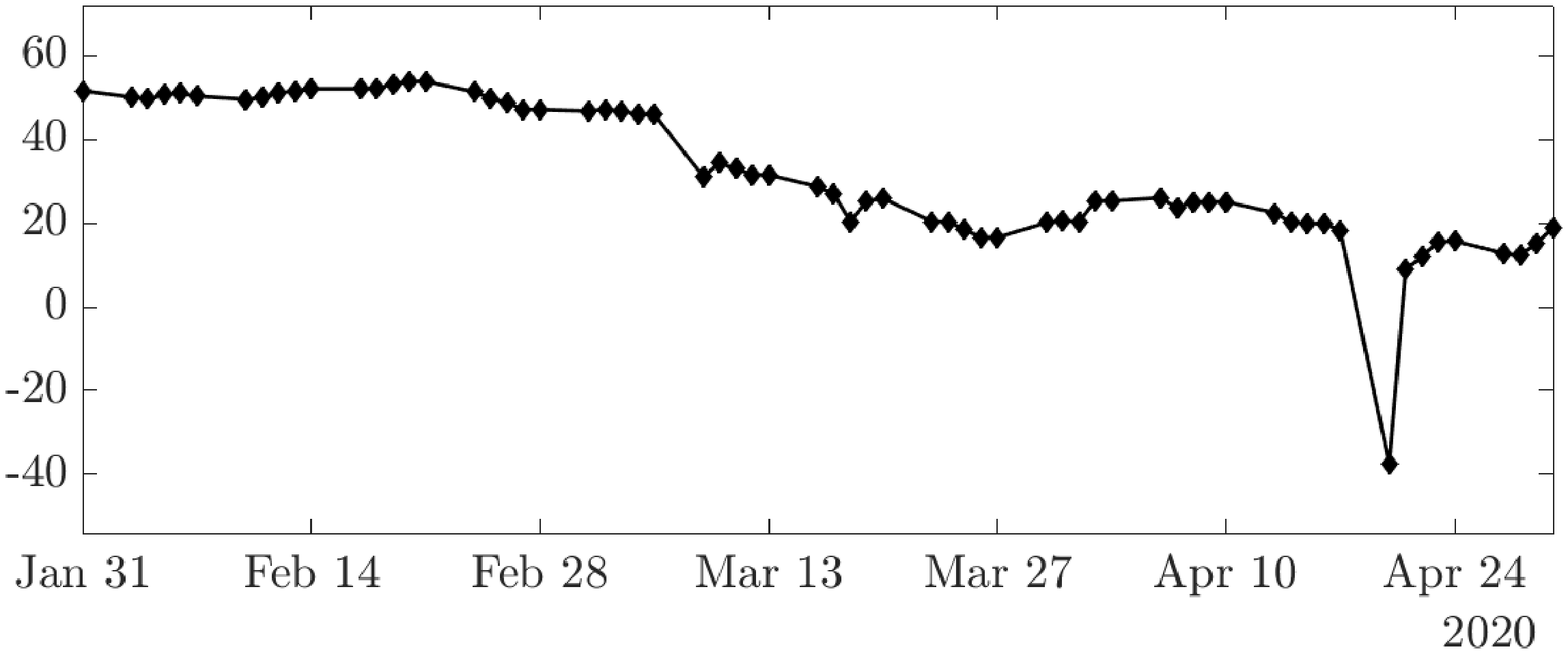} & 
\includegraphics[trim= 20mm 2mm 15mm 5mm,clip,height= 3.5cm, width= 7.5cm]{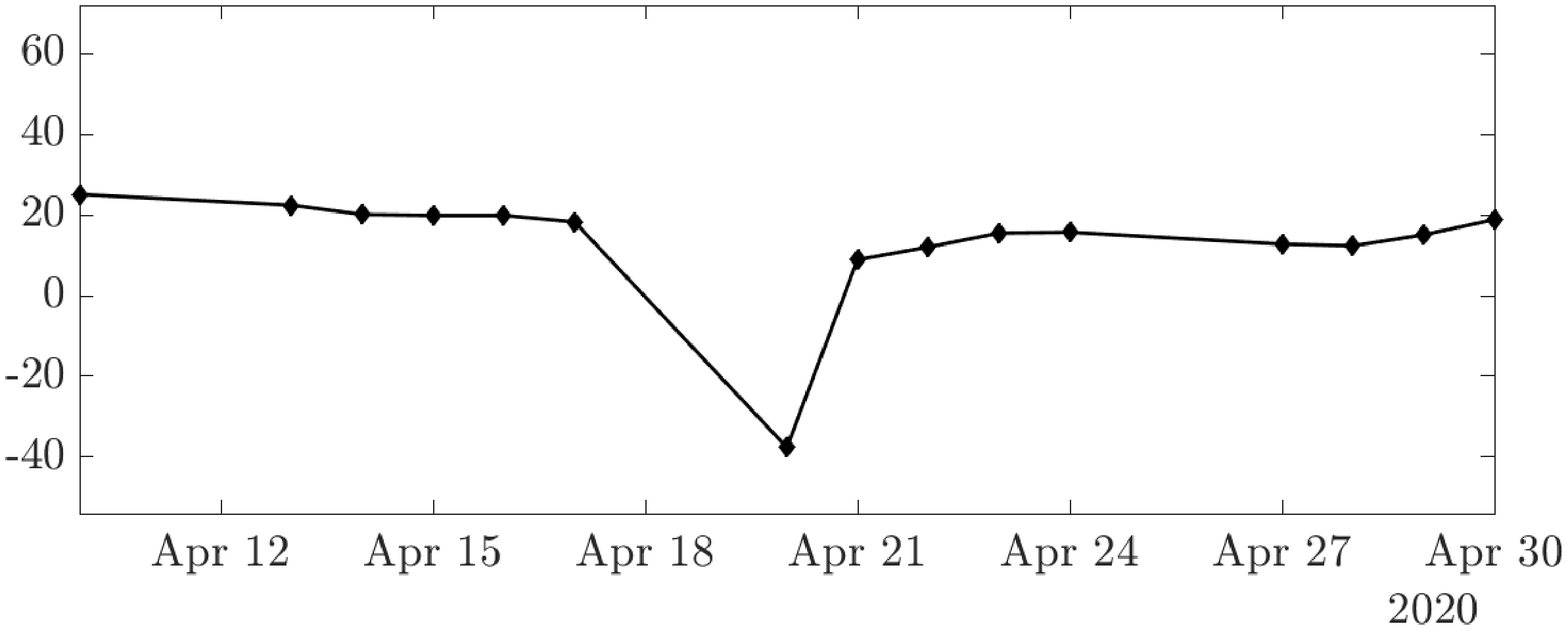}
\end{tabular}
\caption{Top two rows: values of \AS \, for one step ahead density forecasts of the OIL prices according to a TVP-AR(2) model and an AR(20) model, respectively, for selected rolling windows (x-axis) and different asymmetry levels $c$: 0.05 (dashed red line), 0.275 (dashed yellow line), 0.50 (dashed black line), 0.725 (dashed purple line), 0.95 (dashed green line).
Bottom row: observed values of the time series (solid black line).
The right column is a zoomed-in version of the left column.
}
\label{fig:example_score_OIL_AR20}
\end{figure}

These results highlight how accounting for asymmetry in forecast evaluation may lead to dramatically different implications. By looking at the period until April 21, a decision maker averse to overestimation of oil price is likely to discard the AR(20) in favor of alternative models for making forecasts. Conversely, another agent facing the same decision problem, equipped with the same data and models, but averse to underestimation, is likely to agree to the AR(20).

Moreover, these insights provide an important value added of the \AS \, as compared to symmetric scores. By looking at variation of the ranking according to the \AS \, over time, it is possible to infer the relative dynamics of the forecasting and observation densities. In the case previously mentioned, between April 17 and April 21 the forecast tends to overestimate (i.e., its CDF is to the right of the observations CDF), while it tends to underestimate between April 21 and April 24 (i.e., its CDF is to the left of the observations CDF).
Under a symmetric score it is not possible to grasp these insights since negative and positive deviations from the target are equally penalized.

The bottom panel of Tab.~\ref{tab:Real_CRPS_ACPS} reports the results for the electricity prices. As in the previous cases, there is large uncertainty on the model ranking. In line with PIT evidence, the high volatility, spikes and negative prices of the electricity prices drive different results depending on the level of asymmetry of the user. At $h=1$ and $c=0.05$, the AR(20)-tSV is the best model, for higher values of $c$, the TVP-AR(2)-SV and TVP-AR(2)-tSV are the preferred ones. Many models with time-varying volatility outperform the constant volatility models, confirming evidence in \cite{GRR2020b}. At $h=12$ the AR(20) for $c=0.05$, the TVP-AR(2)-tSV for $c=0.5$, and the TVP-AR(1) for $c =0.95$ give the highest \AS. Fig. \ref{fig:EEX_ACPS} again indicates more stable performance of some models when accounting for asymmetry relative to use the symmetric CRPS.

\begin{figure}[!th]
\centering
\setlength{\tabcolsep}{0.2pt}
\setlength{\abovecaptionskip}{-3pt}
\hspace*{-12ex}
\begin{tabular}{cc}
\begin{rotate}{90} {\footnotesize \hspace*{45pt} $h=1$} \end{rotate} &
\includegraphics[trim= 25mm 5mm 40mm 5mm,clip,height= 4.0cm, width= 17.5cm]{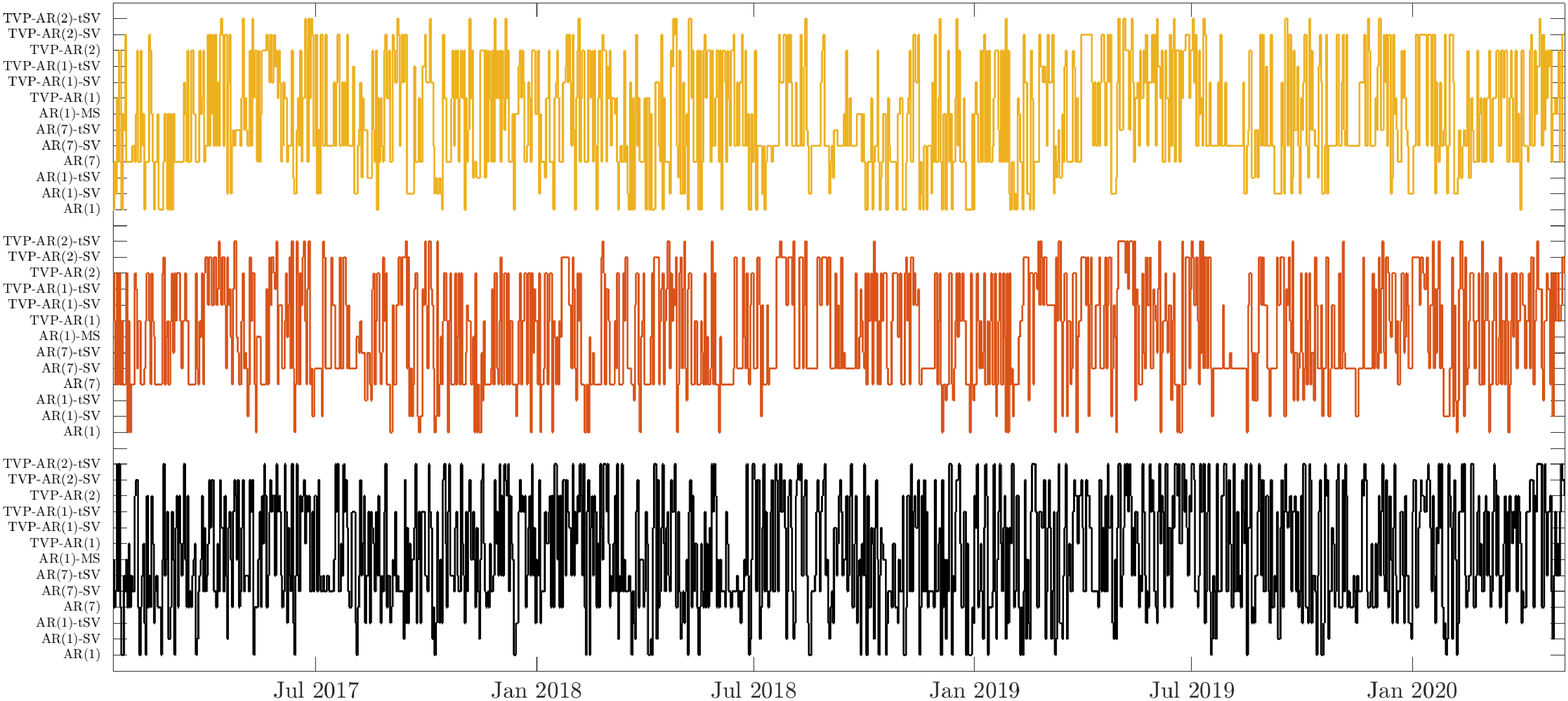} \\[-10pt]
\begin{rotate}{90} {\footnotesize \hspace*{45pt} $h=7$} \end{rotate} &
\includegraphics[trim= 25mm 5mm 40mm 5mm,clip,height= 4.0cm, width= 17.5cm]{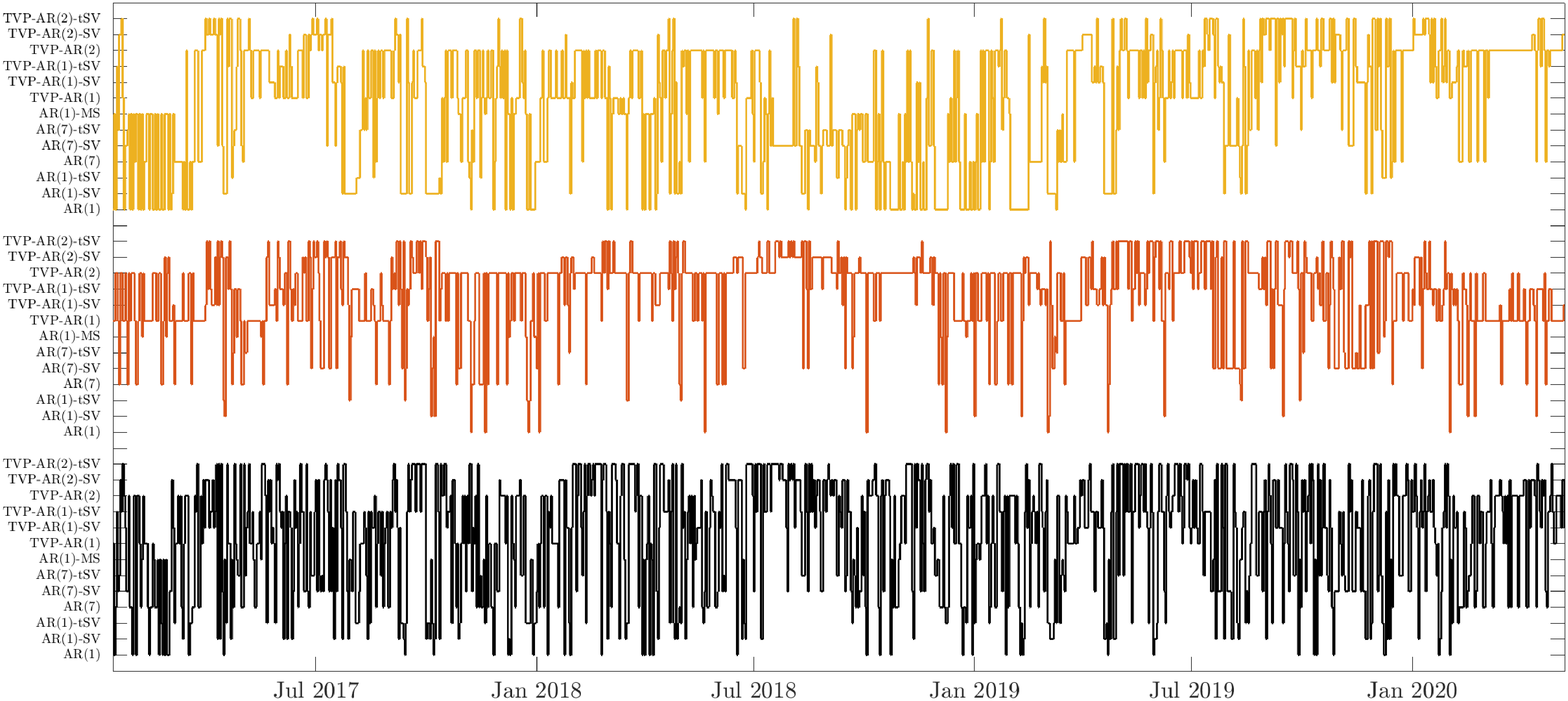}
\end{tabular}
\caption{Best model in each vintage for EEX dataset: CRPS (black), \AS \, with $c=0.05$ (red), \AS \, with $c=0.95$ (yellow).}
\label{fig:EEX_ACPS}
\end{figure}

\section{Conclusions}    \label{sec:conclusions}
This paper has introduced a novel asymmetric proper score for probabilistic forecasts of continuous variables, the \AS. Its main application is the evaluation and comparison of density forecasts.
In addition, we have proposed a threshold- and quantile-weighted version of the asymmetric score, which, by reweighing the domain, allows for a further level of asymmetry in the evaluation of forecasts. We also derive a test to compare the statistical accuracy of different forecasts. The definition of \AS \, is sufficiently flexible to be used in a variety of univariate contexts, and carries over to the multivariate case. The latter deserves further investigation and is an open field for future research.

We provide a tool able to account for the decision maker's preferences in the evaluation of density forecasts both in terms of domain- and error-weighting schemes. 

In an artificial data exercise, we have shown the good performance of our proposed asymmetric score for different continuous target distributions. In relevant macroeconomic and energy applications, we evaluate our score across different models and for different horizons, and we improve on the quality of the forecasts by providing an effective tool for density forecast  evaluation.

The proposed score, \AS, is of general use in any situation where the decision maker has asymmetric preferences in the evaluation of forecasts and thus it can be applied to a much wide range of applications. Further extensions could cover the area of forecast instability, see \cite{GiacominiRossi2010}, and the case of state-dependent function of economic variables such as in \cite{ORS2020}.
%

\bibliographystyle{chicago} 
\bibliography{Asymmetry_biblio.bib}

\clearpage

\renewcommand{\thesection}{S.\arabic{section}}
\renewcommand{\theequation}{S.\arabic{equation}}
\renewcommand{\thefigure}{S.\arabic{figure}}
\renewcommand{\thetable}{S.\arabic{table}}
\setcounter{section}{0}
\setcounter{table}{0}
\setcounter{figure}{0}
\setcounter{equation}{0}
\setcounter{page}{1}

\section*{Supplementary Material to ``Proper scoring rules for evaluating asymmetry in density forecasting''}

\begin{abstract}
This appendix contains the details of the model specifications used in the empirical applications in Section~\ref{S_sec:apdx_models}.
The data used in the empirical application is described in Section~\ref{S_sec:apdx_data}, while Section~\ref{S_sec:apdx_application} provides additional details of the empirical analysis.
\end{abstract}


\section{Models}   \label{S_sec:apdx_models}
Let $i \in \{ EMPL, OIL, EEX \}$.

\subsection*{AR(1)}
As a benchmark, we consider a standard autoregressive model of order 1, AR(1), defined as
\begin{equation}
y_{i,t} = \alpha_i + \beta_i y_{i,t-1} + \epsilon_{i,t}, \qquad \epsilon_{i,t} \sim \mathcal{N}(0,\sigma_i^2),
\label{eq:AR1}
\end{equation}
where $t=1,\ldots,T$. As for the other models, we also consider different lag specification.

\subsection*{Markov Switching AR(1)}
The assumption of constant parameters in time series analysis is likely to be incorrect, especially for long time series. Motivated by this fact, we estimate a Markov switching AR(1) model, called MS-AR(1), where the autoregressive parameter jumps according to a latent Markov chain
\begin{equation}
y_{i,t} = \alpha_{i,S_t} + \beta_{i,S_{i,t}} y_{i,t-1} + \epsilon_{i,t}, \qquad \epsilon_{i,t} \sim \mathcal{N}(0,\sigma_{i,S_t}^2)
\label{eq:MS_AR1}
\end{equation}
where each $S_{i,t}$ follows a homogeneous Markov chain with transition matrix
\begin{equation*}
\Xi_i = \begin{bmatrix}
\xi_{i,1,1} & \xi_{i,1,2} \\
\xi_{i,2,1} & \xi_{i,2,2} 
\end{bmatrix}
\end{equation*}
with $\xi_{i,k,j} = P(S_{i,t} = j | S_{i,t} = k)$.

\subsubsection*{Bayesian inference}
We sample the path of $S_{i,t}$ using the forward filter, backward sampler of \cite{fruhwirth2006finite}.

Fix $i$ and define $\boldsymbol{\beta}_{S_t} = (\alpha_{i}, \beta_{i,S_t})'$ and $X_t = (1,y_{t-1})$. Finally, for each state $m=1,\ldots,M$ let $\mathcal{T}_m = \{ t=1,\ldots,T : s_t  = m \}$.
Assume a Gaussian prior for $\boldsymbol{\beta}_m$, for each state $m=1,\ldots,M$, as
\begin{align*}
P(\boldsymbol{\beta}_m) = \mathcal{N}(\underline{\mu}_\beta, \underline{V}_\beta).
\end{align*}
Then, the posterior is obtained as
\begin{align*}
\boldsymbol{\beta}_m | \mathbf{y}, \mathbf{s}_t, \sigma^2 \sim  \mathcal{N}(\overline{\mu}_\beta, \overline{V}_\beta),
\end{align*}
where
\begin{align*}
\overline{V}_\beta = \Big( \underline{V}_\beta^{-1} + \sum_{t\in \mathcal{T}_m} \frac{X_t' X_t}{\sigma_m^2} \Big)^{-1},   \qquad
\overline{\mu}_\beta = \overline{V}_\beta \Big( \underline{V}_\beta^{-1} \underline{\mu}_\beta + \sum_{t\in \mathcal{T}_m} \frac{X_t' y_t}{\sigma_m^2} \Big).
\end{align*}

Assume a Dirichlet prior for $\boldsymbol{\xi}_{m} = (\xi_{m,1},\ldots,\xi_{m,m})$, for each state $m=1,\ldots,M$, as
\begin{align*}
P(\boldsymbol{\xi}_m) = \mathcal{D}ir(\underline{\mathbf{c}}_m).
\end{align*}
Then, the posterior is obtained as
\begin{align*}
\boldsymbol{\xi}_m | \mathbf{s}_t \sim \mathcal{D}ir(\underline{\mathbf{c}}_m + \mathbf{N}_m),
\end{align*}
where
\begin{align*}
N_{m,l} = \#\{ t : S_{t-1} = m \wedge S_t = l \}.
\end{align*}
For identifying the states, we impose a restriction on the levels of the volatility, i.e. $\sigma_{i,1}^2 < \sigma_{i,2}^2$. Therefore, we label state 1 the low volatility regime and state 2 the high volatility regime.

\subsection*{TVP-AR(1)}
In order to capture potential smooth variations of the coefficient, we consider a TVP framework, denoted TVP-AR(1), where the latter parameter is assumed to evolve according to a latent AR(1) process
\begin{alignat}{3}
& y_{i,t} && = \alpha_{i,t} + \beta_{i,t} y_{i,t-1} + \epsilon_{i,t}, && \qquad \epsilon_{i,t} \sim \mathcal{N}(0,\sigma_i^2) \\
& \boldsymbol{\beta}_{i,t} && = A_i \boldsymbol{\beta}_{i,t-1} + \boldsymbol{\eta}_{i,t}, && \qquad \boldsymbol{\eta}_{i,t} \sim \mathcal{N}(0,\Omega_i)
\label{eq:TVP_AR1}
\end{alignat}
We sample the path of $\beta_{i,t}$ using the forward filter, backward sampler of \cite{Carter1994Gibbs_StateSpace}.

We assume a Gaussian prior for each column $j$ of $A_i$, that is
\begin{equation*}
P(A_{i,:j}) = \mathcal{N}(\underline{\boldsymbol{\mu}}_A, \underline{V}_A),
\end{equation*}
thus obtaining a Gaussian posterior
\begin{equation*}
A_{i,:j} | \{ \boldsymbol{\beta}_{i,1},\ldots,\boldsymbol{\beta}_{i,T} \}, \Omega_i \sim \mathcal{N}(\overline{\boldsymbol{\mu}}_A, \overline{V}_A),
\end{equation*}
where
\begin{align*}
\hspace*{-6ex}
\overline{V}_A = \Big( \underline{V}_A^{-1} + \sum_{t=2}^T \boldsymbol{\beta}_{i,t-1}' \Omega_i^{-1} \boldsymbol{\beta}_{i,t-1} \Big)^{-1},   \qquad
\overline{\boldsymbol{\mu}}_A = \overline{V}_A \Big( \underline{V}_A^{-1} \underline{\boldsymbol{\mu}}_A + \sum_{t=2}^T \boldsymbol{\beta}_{i,t-1}' \Omega_i^{-1} \boldsymbol{\beta}_{i,t} \Big).
\end{align*}


\subsection*{Stochastic Volatility}
We also study the previous autoregressive models with the inclusion of stochastic volatility. We present only the extension for the AR(1) in eq.~\eqref{eq:AR1}, and analogous changes apply to model \eqref{eq:TVP_AR1}.

Let $h_{i,t}$, for $t=1,\ldots,T$, be the log-variance for series $i$ at time $t$, then an AR(1) model with Gaussian noise and stochastic volatility, denoted AR(1)-SV, is given by
\begin{alignat}{3}
\notag
y_{i,t} & = \alpha_i + \beta_i y_{i,t-1} + \epsilon_{i,t},  \qquad & \epsilon_{i,t} & \sim \mathcal{N}(0,\exp(h_{i,t})) \\
\label{eq:AR_SV}
h_{i,t} & = h_{i,t-1} + u_{i,t}, \qquad & u_{i,t} & \sim \mathcal{N}(0,\sigma_{i,h}^2).
\end{alignat}

The last model considered is an extension of the AR(1) with stochastic volatility that assumes Student-$t$ noise. Let $\nu_i > 0$ be the degrees of freedom parameter. An AR(1) with Student-$t$ noise and stochastic volatility, denoted AR(1)-SV-$t$, is given by
\begin{alignat}{3}
\notag
y_{i,t} & = \alpha_i + \beta_i y_{i,t-1} + \epsilon_{i,t},  \qquad & \epsilon_{i,t} & \sim \mathcal{N}(0,\exp(h_{i,t})/\lambda_{i,t}) \\
\label{eq:model_ARX_SV_t}
h_{i,t} & = h_{i,t-1} + u_{i,t}, \qquad & u_{i,t} & \sim \mathcal{N}(0,\sigma_{i,h}^2),
\end{alignat}
where $\lambda_{i,t} \sim \mathcal{IG}(\nu_i/2,\nu_i/2)$.

\clearpage
\section{Data description}   \label{S_sec:apdx_data}

\begin{figure}[H]
\centering
\caption{\bf Raw data series.}
\vspace{-0.25in}
\begin{justify} 
\centering
\footnotesize{This figure reports the raw time series for EMPL, OIL, and EEX .}
\end{justify}
\hspace{-.6em}
\begin{adjustbox}{width=1.1\textwidth,center=\textwidth}
\begin{tabular}{ccc}
\begin{rotate}{90} \scriptsize \hspace*{28pt} growth rate EMPL \end{rotate} & \hspace*{-12pt}
\includegraphics[trim= 12mm 0mm 24mm 0mm,clip,height= 4.0cm, width= 10.7cm]{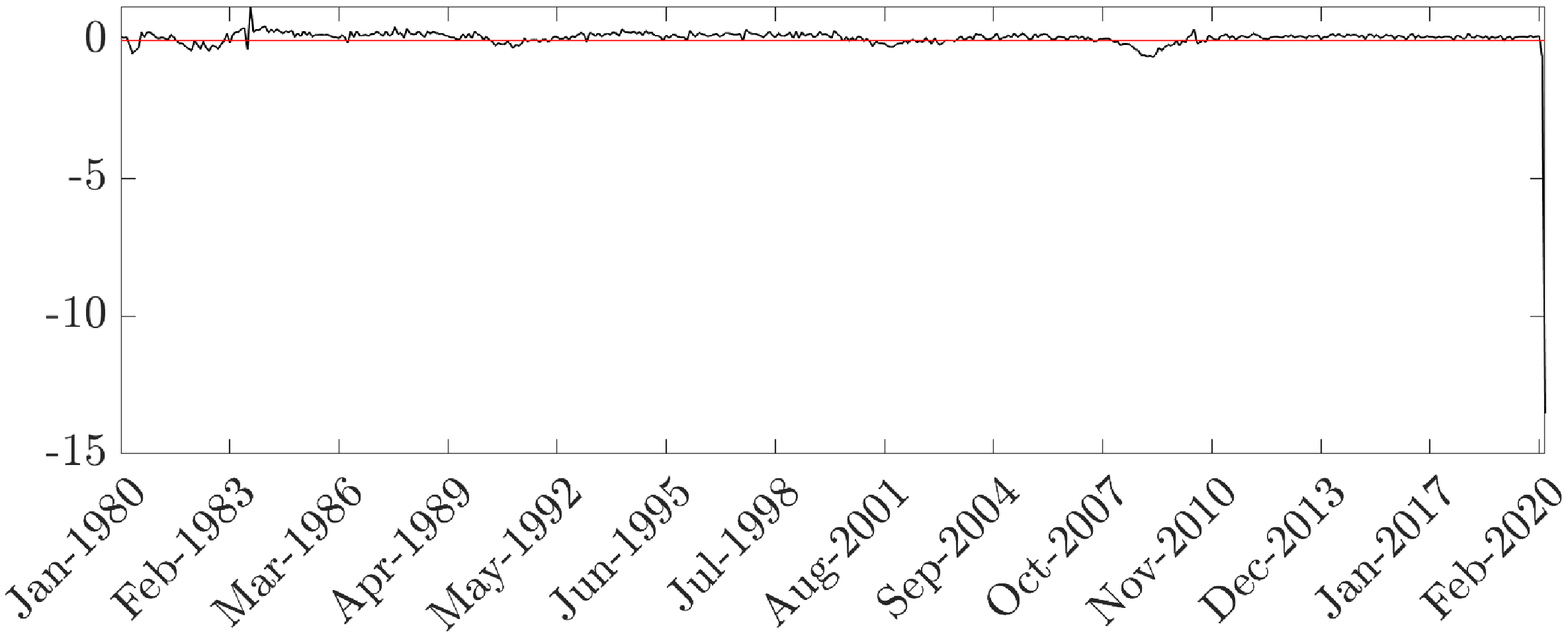} \\
\begin{rotate}{90} \scriptsize \hspace*{60pt} OIL \end{rotate} & \hspace*{-12pt}
\includegraphics[trim= 12mm 0mm 24mm 0mm,clip,height= 4.0cm, width= 10.7cm]{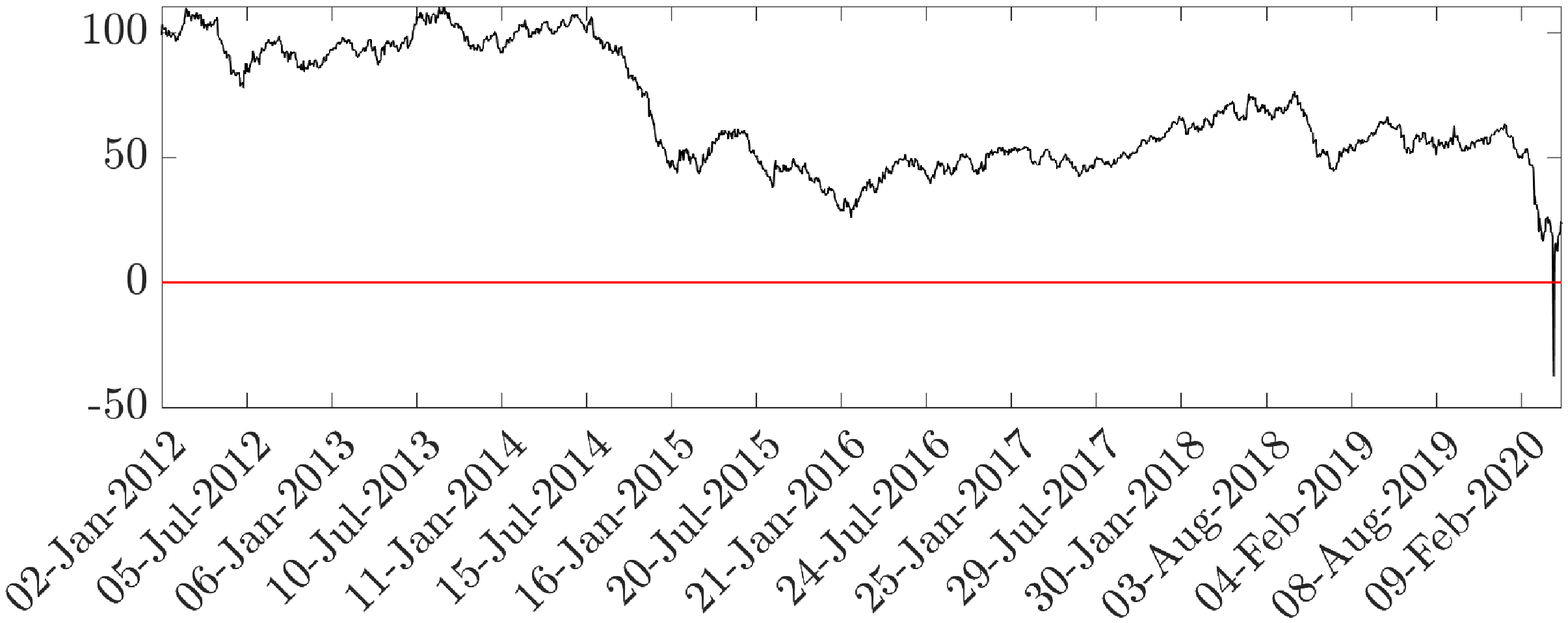} \\
\begin{rotate}{90} \scriptsize \hspace*{60pt} EEX \end{rotate} & \hspace*{-12pt}
\includegraphics[trim= 12mm 0mm 24mm 0mm,clip,height= 4.0cm, width= 10.7cm]{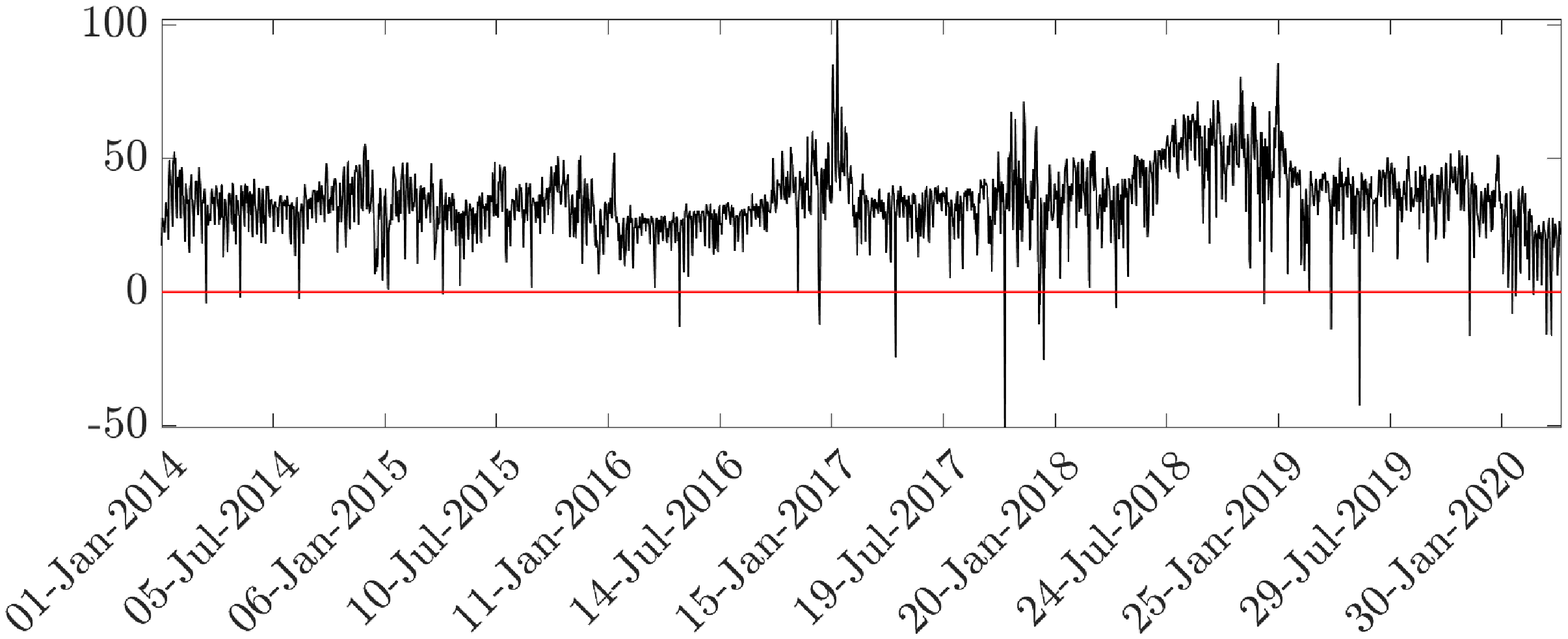}
\end{tabular}
\end{adjustbox}
\label{fig:data_rraw}
\end{figure}

\clearpage

\section{Additional details on empirical applications}   \label{S_sec:apdx_application}

\begin{table}[!th] 
\centering 
\caption{Relative frequency of occurrence of each model for EMPL as best model, across vintages, accoriding to: CRPS (rows 1-2), ACPS with $c=0.05$ (rows 3-4), ACPS with $c=0.95$ (rows 5-6)} 
\begin{adjustbox}{width=1.3\textwidth, center=\textwidth}
\begin{tabular}{|ll|*{13}{c}|} 
\toprule 
&& AR(1)& AR(1)-SV& AR(1)-tSV& AR(12)& AR(12)-SV& AR(12)-tSV& AR(1)-MS& TVP-AR(1)& TVP-AR(1)-SV& TVP-AR(1)-tSV& TVP-AR(2)& TVP-AR(2)-SV& TVP-AR(2)-tSV\\ 
\midrule 
\multirow[]{2}{*}{CRPS} &
$h=1$& 0.01 & 0.10 & 0.10 & 0.03 & 0.16 & 0.13 & 0.00 & 0.06 & 0.07 & 0.07 & 0.10 & 0.09 & 0.06  \\ 
&$h=12$& 0.01 & 0.10 & 0.14 & 0.05 & 0.15 & 0.15 & 0.00 & 0.17 & 0.06 & 0.05 & 0.04 & 0.05 & 0.04  \\ 
\multirow[]{2}{*}{ACPS 0.05} &
$h=1$& 0.00 & 0.20 & 0.05 & 0.00 & 0.31 & 0.06 & 0.01 & 0.10 & 0.04 & 0.03 & 0.11 & 0.05 & 0.03  \\ 
&$h=12$& 0.01 & 0.18 & 0.03 & 0.00 & 0.31 & 0.07 & 0.03 & 0.31 & 0.00 & 0.00 & 0.04 & 0.00 & 0.00  \\ 
\multirow[]{2}{*}{ACPS 0.95} &
$h=1$& 0.00 & 0.16 & 0.03 & 0.00 & 0.26 & 0.08 & 0.00 & 0.15 & 0.07 & 0.03 & 0.16 & 0.03 & 0.03  \\ 
&$h=12$& 0.00 & 0.19 & 0.05 & 0.00 & 0.23 & 0.04 & 0.00 & 0.45 & 0.00 & 0.00 & 0.05 & 0.00 & 0.00  \\ 
\bottomrule 
\end{tabular}
\end{adjustbox}
\end{table}

\begin{table}[!th] 
\centering 
\caption{Relative frequency of occurrence of each model for OIL as best model, across vintages, according to: CRPS (rows 1-2), ACPS with $c=0.05$ (rows 3-4), ACPS with $c=0.95$ (rows 5-6)}
\begin{adjustbox}{width=1.3\textwidth, center=\textwidth}
\begin{tabular}{|ll|*{13}{c}|} 
\toprule 
& & AR(1)& AR(1)-SV& AR(1)-tSV& AR(20)& AR(20)-SV& AR(20)-tSV& AR(1)-MS& TVP-AR(1)& TVP-AR(1)-SV& TVP-AR(1)-tSV& TVP-AR(2)& TVP-AR(2)-SV& TVP-AR(2)-tSV\\ 
\midrule 
\multirow[]{2}{*}{CRPS} &
$h=1$& 0.03 & 0.04 & 0.06 & 0.07 & 0.04 & 0.06 & 0.10 & 0.04 & 0.06 & 0.05 & 0.11 & 0.17 & 0.18  \\ 
&$h=5$& 0.06 & 0.05 & 0.08 & 0.11 & 0.09 & 0.11 & 0.09 & 0.07 & 0.07 & 0.07 & 0.05 & 0.07 & 0.08  \\ 
\multirow[]{2}{*}{ACPS 0.05} &
$h=1$& 0.03 & 0.11 & 0.05 & 0.04 & 0.09 & 0.05 & 0.12 & 0.05 & 0.06 & 0.05 & 0.13 & 0.11 & 0.10  \\ 
&$h=5$& 0.03 & 0.08 & 0.07 & 0.10 & 0.19 & 0.12 & 0.14 & 0.09 & 0.04 & 0.04 & 0.05 & 0.03 & 0.02  \\ 
\multirow[]{2}{*}{ACPS 0.95} &
$h=1$& 0.04 & 0.12 & 0.05 & 0.05 & 0.10 & 0.05 & 0.11 & 0.03 & 0.06 & 0.04 & 0.12 & 0.12 & 0.10  \\ 
&$h=5$& 0.05 & 0.13 & 0.05 & 0.16 & 0.19 & 0.09 & 0.12 & 0.07 & 0.03 & 0.02 & 0.04 & 0.02 & 0.02  \\ 
\bottomrule 
\end{tabular} 
\end{adjustbox}
\end{table}

\begin{table}[!th] 
\centering 
\caption{Relative frequency of occurrence of each model for EEX as best model, across vintages, according to: CRPS (rows 1-2), ACPS with $c=0.05$ (rows 3-4), ACPS with $c=0.95$ (rows 5-6)} 
\begin{adjustbox}{width=1.3\textwidth, center=\textwidth}
\begin{tabular}{|ll|*{13}{c}|} 
\toprule 
&& AR(1)& AR(1)-SV& AR(1)-tSV& AR(7)& AR(7)-SV& AR(7)-tSV& AR(1)-MS& TVP-AR(1)& TVP-AR(1)-SV& TVP-AR(1)-tSV& TVP-AR(2)& TVP-AR(2)-SV& TVP-AR(2)-tSV\\ 
\midrule 
\multirow[]{2}{*}{CRPS} &
$h=1$& 0.04 & 0.03 & 0.05 & 0.11 & 0.16 & 0.10 & 0.04 & 0.05 & 0.06 & 0.11 & 0.08 & 0.07 & 0.10  \\ 
& $h=7$& 0.04 & 0.05 & 0.05 & 0.07 & 0.07 & 0.04 & 0.03 & 0.08 & 0.08 & 0.12 & 0.12 & 0.10 & 0.16  \\ 
\multirow[]{2}{*}{ACPS 0.05} &
$h=1$& 0.03 & 0.03 & 0.02 & 0.18 & 0.18 & 0.04 & 0.03 & 0.09 & 0.08 & 0.04 & 0.13 & 0.11 & 0.04  \\ 
& $h=7$& 0.01 & 0.01 & 0.01 & 0.04 & 0.05 & 0.01 & 0.01 & 0.18 & 0.05 & 0.09 & 0.32 & 0.10 & 0.12  \\ 
\multirow[]{2}{*}{ACPS 0.95} &
$h=1$& 0.05 & 0.05 & 0.03 & 0.16 & 0.19 & 0.04 & 0.05 & 0.08 & 0.07 & 0.02 & 0.12 & 0.10 & 0.03  \\ 
& $h=7$& 0.09 & 0.06 & 0.01 & 0.08 & 0.06 & 0.03 & 0.05 & 0.12 & 0.05 & 0.04 & 0.27 & 0.06 & 0.08  \\ 
\bottomrule 
\end{tabular} 
\end{adjustbox}
\end{table}

\begin{table}[!th] 
\centering 
\caption{Ranking of probability forecasts and accuracy test. Best model, over vintages, according to: CRPS; ACPS with $c=0.05;0.275; 0.5; 0.725; 0.95$ for the three different datasets: Employment (top); Oil (middle) and EEX (bottom).} 
\begin{adjustbox}{width=1.3\textwidth, center=\textwidth}
\begin{threeparttable}
\begin{tabular}{l*{13}{c}} 
\toprule 
& & & &  & & & \multicolumn{2}{c}{\textsc{EMPL}} & \\
\midrule 
\multicolumn{1}{l}{\textit{Horizon 1}} & \cg AR(1)& \cg AR(1)-SV& \cg  AR(1)-tSV& \cg AR(12)& \cg AR(12)-SV& \cg AR(12)-tSV& \cg AR(1)-MS& \cg TVP-AR(1)& \cg TVP-AR(1)-SV& \cg TVP-AR(1)-tSV& \cg TVP-AR(2)& \cg TVP-AR(2)-SV& \cg TVP-AR(2)-tSV \\ 
\AS$(\cdot,\cdot;0.05)$ & 
12 & 9\trestar & 5\trestar & 13 & 2\trestar & 1\trestar & 11\trestar & 3\trestar & 10\trestar & 8\trestar & 4\trestar & 7\trestar & 6\trestar \\
\AS$(\cdot,\cdot;0.275)$ & 12 & 10\trestar & 9\trestar & 13 & 8\trestar & 7\trestar & 11\trestar & 2\trestar & 6\trestar & 5\trestar & 1\trestar & 3\trestar & 4\trestar \\
\AS$(\cdot,\cdot;0.5)$ & 13 & 10\trestar & 9\trestar & 11 & 8\trestar & 7\trestar & 12\trestar & 4\trestar & 6\trestar & 5\trestar & 1\trestar & 2\trestar & 3\trestar \\
\AS$(\cdot,\cdot;0.725)$ & 13 & 10\trestar & 9\trestar & 11\duestar & 2\trestar & 1\trestar & 12\trestar & 8\trestar & 6\trestar & 7\trestar & 4\trestar & 3\trestar & 5\trestar \\
\AS$(\cdot,\cdot;0.95)$   & 13 & 4\trestar & 3\trestar & 12 & 1\trestar & 2\trestar & 11\trestar & 5\trestar & 7\trestar & 9\trestar & 6\trestar & 8\trestar & 10\trestar \\
CRPS & 13 & 10\trestar & 9\trestar & 11 & 8\trestar & 7\trestar & 12\trestar & 4\trestar & 6\trestar & 5\trestar & 1\trestar & 2\trestar & 3\trestar \\
Knuppel p-value & 0.137 & 0.389 & 0.374 & 0.475 & 0.222 & 0.180 & 0.145 & 0.337 & 0.136 & 0.136 & 0.195 & 0.061 & 0.061 \\
\midrule
\multicolumn{1}{l}{\textit{Horizon 12}}  & \cg AR(1)& \cg AR(1)-SV& \cg AR(1)-tSV& \cg AR(12)& \cg AR(12)-SV& \cg AR(12)-tSV& \cg AR(1)-MS& \cg TVP-AR(1)& \cg TVP-AR(1)-SV& \cg TVP-AR(1)-tSV& TVP-AR(2)& \cg TVP-AR(2)-SV& \cg TVP-AR(2)-tSV\\ 
\AS$(\cdot,\cdot;0.05)$ & 11 & 12 & 10 & 13 & 3 & 2 & 9\trestar & 4 & 6 & 5 & 1 & 8 & 7 \\
\AS$(\cdot,\cdot;0.275)$ & 12 & 8\trestar & 5\trestar & 13 & 1\trestar & 2\trestar & 11\trestar & 3\trestar & 7\trestar & 6\trestar & 4\trestar & 9\trestar & 10\trestar \\
\AS$(\cdot,\cdot;0.5)$ & 12 & 4\trestar & 3\trestar & 13 & 1\trestar & 2\trestar & 11\trestar & 5\trestar & 7\trestar & 8\trestar & 6\trestar & 9\trestar & 10\trestar \\
\AS$(\cdot,\cdot;0.725)$ & 12 & 4\trestar & 1\trestar & 13 & 2\trestar & 3\trestar & 11\trestar & 5\trestar & 6\trestar & 7\trestar & 8\trestar & 9\trestar & 10\trestar \\
\AS$(\cdot,\cdot;0.95)$   & 12 & 1\trestar & 2\trestar & 13 & 3\trestar & 4\trestar & 11\trestar & 5\trestar & 7\trestar & 8\trestar & 6\trestar & 9\trestar & 10\trestar \\
CRPS & 12 & 4\trestar & 3\trestar & 13 & 1\trestar & 2\trestar & 11\trestar & 5\trestar & 7\trestar & 8\trestar & 6\trestar & 9\trestar & 10\trestar \\
Knuppel p-value & 0.181 & 0.088 & 0.126 & 0.440 & 0.124 & 0.138 & 0.170 & 0.087 & 0.354 & 0.354 & 0.049 & 0.287 & 0.287 \\
\bottomrule 
\toprule 
& & & &  & & & \multicolumn{2}{c}{\textsc{OIL}} & \\
\midrule 
\multicolumn{1}{l}{\textit{Horizon 1}}  & AR(1)& AR(1)-SV& AR(1)-tSV& AR(20)& AR(20)-SV& AR(20)-tSV& AR(1)-MS& TVP-AR(1)& TVP-AR(1)-SV& TVP-AR(1)-tSV& TVP-AR(2)& \cg TVP-AR(2)-SV& \cg TVP-AR(2)-tSV \\  
\AS$(\cdot,\cdot;0.05)$ & 12 & 8\unastar & 11 & 10\duestar & 7\unastar & 9\duestar & 13 & 6\trestar & 3\trestar & 4\trestar & 5\trestar & 2\trestar & 1\trestar \\
\AS$(\cdot,\cdot;0.275)$ & 12 & 8 & 13 & 7\unastar & 9 & 10 & 11 & 4\trestar & 5\trestar & 6\trestar & 1\trestar & 3\trestar & 2\trestar \\
\AS$(\cdot,\cdot;0.5)$ & 12 & 9 & 13 & 7 & 8 & 11 & 10 & 3\trestar & 6\trestar & 5\trestar & 1\trestar & 4\trestar & 2\trestar \\
\AS$(\cdot,\cdot;0.725)$ & 10 & 9 & 13 & 7 & 8 & 11 & 12 & 2\trestar & 6\duestar & 5\duestar & 1\trestar & 4\trestar & 3\trestar \\
\AS$(\cdot,\cdot;0.95)$   & 12 & 8 & 11\duestar & 9 & 7 & 10\duestar & 13 & 3\duestar & 6\unastar & 5\unastar & 1\duestar & 4\unastar & 2\unastar \\
CRPS & 12 & 9 & 13 & 7 & 8 & 11 & 10 & 3\trestar & 6\trestar & 5\trestar & 1\trestar & 4\trestar & 2\trestar \\
Knuppel p-value & 0.000 & 0.000 & 0.000 & 0.000 & 0.000 & 0.000 & 0.000 & 0.011 & 0.042 & 0.042 & 0.001 & 0.159 & 0.159 \\
\midrule
\multicolumn{1}{l}{\textit{Horizon 5}} & AR(1)& AR(1)-SV& AR(1)-tSV& AR(20)& AR(20)-SV& AR(20)-tSV& AR(1)-MS& TVP-AR(1)& TVP-AR(1)-SV& TVP-AR(1)-tSV& TVP-AR(2)& TVP-AR(2)-SV& TVP-AR(2)-tSV\\  
\AS$(\cdot,\cdot;0.05)$ & 10 & 5 & 7 & 9 & 3 & 6 & 12 & 13 & 2 & 8 & 11 & 1\unastar & 4\unastar \\
\AS$(\cdot,\cdot;0.275)$ & 5 & 8 & 7 & 6 & 4 & 3 & 11 & 1\duestar & 13 & 12 & 2\unastar & 10 & 9 \\
\AS$(\cdot,\cdot;0.5)$ & 3 & 7 & 6 & 8 & 4 & 5 & 11 & 1\duestar & 13 & 12 & 2 & 10 & 9 \\
\AS$(\cdot,\cdot;0.725)$ & 2 & 4 & 5 & 8 & 3 & 6 & 9 & 1 & 13 & 11 & 7 & 12 & 10 \\
\AS$(\cdot,\cdot;0.95)$   & 6 & 1 & 3\unastar  & 8 & 2 & 4 & 13 & 5 & 10 & 7 & 11 & 12 & 9 \\
CRPS & 3 & 6 & 7 & 8 & 4 & 5 & 11 & 1\duestar & 13 & 12 & 2 & 10 & 9 \\
Knuppel p-value & 0.025 & 0.021 & 0.012 & 0.040 & 0.016 & 0.009 & 0.004 & 0.002 & 0.017 & 0.017 & 0.001 & 0.002 & 0.002 \\
\bottomrule 
\toprule 
& & & &  & & & \multicolumn{2}{c}{\textsc{EEX}} & \\
\midrule 
\multicolumn{1}{l}{\textit{Horizon 1}}  & AR(1)& AR(1)-SV& AR(1)-tSV& AR(20)& AR(20)-SV& AR(20)-tSV& AR(1)-MS& TVP-AR(1)& TVP-AR(1)-SV& TVP-AR(1)-tSV& TVP-AR(2)& TVP-AR(2)-SV& TVP-AR(2)-tSV\\ 
\AS$(\cdot,\cdot;0.05)$ &11 & 8\unastar & 6\trestar & 9\trestar & 3\trestar & 1\trestar & 12 & 13 & 7\duestar & 4\duestar & 10 & 5\duestar & 2\trestar \\
\AS$(\cdot,\cdot;0.275)$ & 12 & 10 & 11 & 5\trestar & 2\trestar & 1\trestar & 13 & 9\trestar & 7\trestar & 6\trestar & 8\trestar & 4\trestar & 3\trestar \\
\AS$(\cdot,\cdot;0.5)$ & 12 & 10\trestar & 11\unastar & 7\trestar & 2\trestar & 5\trestar & 13 & 9\trestar & 6\trestar & 4\trestar & 8\trestar & 3\trestar & 1\trestar \\
\AS$(\cdot,\cdot;0.725)$ & 12 & 10\trestar & 11\trestar & 9\trestar & 5\trestar & 6\trestar & 13 & 8\trestar & 3\trestar & 2\trestar & 7\trestar & 1\trestar & 4\trestar \\
\AS$(\cdot,\cdot;0.95)$   & 12 & 7\trestar & 8\trestar & 11\trestar & 2\trestar & 6\trestar & 13 & 9\trestar & 3\trestar & 5\trestar & 10\trestar & 1\trestar & 4\trestar \\
CRPS & 13 & 10\trestar & 11\unastar & 7\trestar & 1\trestar & 5\trestar & 12 & 9\trestar & 6\trestar & 4\trestar &  8\trestar & 3\trestar & 2\trestar \\
Knuppel p-value & 0.000 & 0.000 & 0.000 & 0.000 & 0.000 & 0.000 & 0.000 & 0.000 & 0.000 & 0.000 & 0.000 & 0.000 & 0.000 \\
\midrule
\multicolumn{1}{l}{\textit{Horizon 7}}  & AR(1)& AR(1)-SV& AR(1)-tSV& AR(20)& AR(20)-SV& AR(20)-tSV& AR(1)-MS& TVP-AR(1)& TVP-AR(1)-SV& TVP-AR(1)-tSV& TVP-AR(2)& TVP-AR(2)-SV& TVP-AR(2)-tSV\\  
\AS$(\cdot,\cdot;0.05)$ & 5 & 13 & 12 & 1\trestar & 11 & 4 & 7 & 3 & 9 & 8 & 2 & 10 & 6 \\
\AS$(\cdot,\cdot;0.275)$ & 10 & 12 & 13 & 8\trestar & 9\trestar & 7\trestar & 11 & 6\trestar & 5\trestar & 2\trestar & 3\trestar & 4\trestar & 1\trestar \\
\AS$(\cdot,\cdot;0.5)$ & 10 & 12 & 13 & 9\trestar & 8\trestar & 7\trestar & 11 & 6\trestar & 4\trestar & 2\trestar & 5\trestar & 3\trestar & 1\trestar \\
\AS$(\cdot,\cdot;0.725)$ & 12 & 11\duestar & 10\trestar & 9\trestar & 8\trestar & 7\trestar & 13 & 5\trestar & 6\trestar & 2\trestar & 3\trestar & 4\trestar & 1\trestar \\
\AS$(\cdot,\cdot;0.95)$   & 12 & 11 & 10 & 9\trestar & 7\unastar & 5\duestar & 13 & 1\trestar & 8\unastar & 4\duestar & 2\trestar & 6\unastar & 3\duestar \\
CRPS & 10 & 12 & 13 & 9\trestar & 8\trestar & 7\trestar & 11 & 6\trestar & 4\trestar & 2\trestar & 5\trestar & 3\trestar & 1\trestar \\
Knuppel p-value & 0.000 & 0.000 & 0.000 & 0.000 & 0.000 & 0.000 & 0.000 & 0.000 & 0.000 & 0.000 & 0.000 & 0.000 & 0.000 \\
\bottomrule 
\end{tabular}
\small{
\begin{tablenotes}
\item \textit{Notes:}
\item[1] $^{\ast \ast \ast}$, $^{\ast \ast}$ and $^{\ast}$ indicate scores are significantly different from 1 at $1\%$, $5\%$ and $10\%$, according to the ACPStest in Section~\ref{sec:test}.
\item[2] Gray cells indicate models that are correctly
calibrated at $5\%$ significance level according to the Knuppel test. 
\end{tablenotes}
}
\end{threeparttable}
\end{adjustbox}
\label{tab:Real_CRPS_ACPS_Supp}
\end{table}

\end{document}